\documentclass[10pt,journal]{IEEEtran}
\usepackage{cite}
\usepackage{float}
\usepackage{graphicx}
\usepackage{amsmath}
\usepackage{times}
\usepackage{graphicx}
\usepackage{latexsym}
\usepackage{graphicx}
\usepackage{bm}
\usepackage{amssymb}
\usepackage[center]{caption2}
\usepackage{stfloats}
\usepackage{cases}
\usepackage{array}
\usepackage{setspace}
\usepackage{fancyhdr}
\usepackage{color}
\usepackage{subfigure}
\usepackage{hyperref}
\usepackage{multirow}
\usepackage{amsthm}
\usepackage{cases}
\usepackage{algpseudocode}
\usepackage[linesnumbered,ruled,vlined]{algorithm2e}
\usepackage[long]{optidef}
\usepackage{epstopdf}
\usepackage{verbatim}
\usepackage{xcolor}
\usepackage{enumitem}
\allowdisplaybreaks[4]
\def\bm#1{\mbox{\boldmath $#1$}}

\theoremstyle{definition}

\newtheorem{proposition}{Proposition}

\newcommand{\black}[1]{{\textcolor[rgb]{0,0,0}{#1}}}

\begin{document}

\title{Detection and Multi-Parameter Estimation for NLOS Targets: An IRS-assisted Framework}
\author{\IEEEauthorblockN{Zhouyuan Yu, Xiaoling Hu, {\em Member, IEEE},   Chenxi Liu, {\em Senior Member, IEEE}, Qin Tao, {\em Member, IEEE}, and Mugen Peng, {\em Fellow, IEEE}}

\thanks{Zhouyuan Yu, Xiaoling Hu, Chenxi Liu and Mugen Peng  are with the State Key Laboratory of Networking and Switching Technology, Beijing University of Posts and
Telecommunications, Beijing, 100876, China (e-mail: \{zhouyuanyu, xiaolinghu, chenxi.liu, pmg\}@bupt.edu.cn).}
\thanks{Qin Tao is with the School of Information Science and Technology, Hangzhou Normal University, Hangzhou, 311121, China (email: taoqin@hznu.edu.cn).}
\thanks{This work was presented in part at the IEEE WCSP 2023 \cite{yu2023blind}}
\vspace{-1mm}
}


\maketitle

\begin{abstract}
Intelligent reflecting surface (IRS) has the potential to enhance sensing performance, due to its capability of reshaping the echo signals. Different from the existing literature, which has commonly focused on IRS beamforming optimization, in this paper, we pay special attention to designing effective signal processing approaches to extract sensing information from IRS-reshaped echo signals. To this end, we investigate an IRS-assisted non-line-of-sight (NLOS) target detection and multi-parameter estimation problem in orthogonal frequency division multiplexing (OFDM) systems. \black{To address this problem, we first propose a novel detection and direction estimation framework, including a low-overhead hierarchical codebook that allows the IRS to generate three-dimensional beams with adjustable beam direction and width, a delay spectrum peak-based beam training scheme for detection and direction estimation, and a beam refinement scheme for further enhancing the accuracy of the direction estimation.} Then, we propose a target range and velocity estimation scheme by extracting the delay-Doppler information from the IRS-reshaped echo signals. \black{Numerical results demonstrate that the proposed schemes can achieve $99.7\%$ target detection rate, a $10^{-3}$-rad level direction estimation accuracy, and a $10^{-6}$-m/$10^{-5}$-m/s level range/velocity estimation accuracy.}

\end{abstract}


\section{Introduction}

Intelligent reflecting surface (IRS) is an emerging technology for future wireless communication networks, which is composed of meta-material units capable of dynamically controlling the phase of the incoming signal to reconfigure the wireless propagation environment\cite{8910627}. Besides, the IRS enables signal manipulation without the need for power amplifiers, which facilitates the reduction in energy consumption and hardware costs.

\par Despite originally intended to enhance communication performance\cite{MeiWeidongIntelligent}, IRS has lately attracted growing interest in sensing applications, such as target detection and parameter estimation. In terms of performance analysis, the work in \cite{BuzziRadar} investigated the potential of the IRS for target detection, and demonstrated the substantial target detection gain attained from the IRS. Later on, the work \cite{BuzziFoundations} analyzed the IRS-assisted target detection performance in an MIMO radar scenario. In \cite{AubryReconfigurable}, the authors addressed the problem of blind-zone radar surveillance with the aid of an IRS, and revealed the significant improvement in signal-to-noise ratio (SNR) with an increased size of IRS. For exploring the IRS's potential for localization, the authors in \cite{HuBeyond} derived the Cramer-Rao lower bound (CRLB) for positioning with IRS, which was verified to decrease quadratically with the surface area of IRS. In addition, Huang \emph{et al}. \cite{HuangJoint} extended the IRS-assisted localization to a multiple-IRS scenario, showing that localization accuracy can be improved by deploying more IRSs.

\black{Owing to the immense potential of IRS in sensing applications, researchers have dedicated substantial efforts to design IRS beamforming in passive IRS-aided sensing systems \cite{SongXianxinIntelligent,esmaeilbeig2023moving, EsmaeilbeigCramérRao, esmaeilbeig2022irs, ZhangMetaLocalization, AhmadRIS}. 
For example, Esmaeilbeig \emph{et al}. designed IRS phase shifts to enhance the performance of target detection \cite{esmaeilbeig2023moving} and target parameter estimation \cite{EsmaeilbeigCramérRao,esmaeilbeig2022irs}, demonstrating that the target sensing performance can be remarkably improved by appropriate IRS beamforming design. For enhancing multi-user localization performance, authors in \cite{ZhangMetaLocalization} jointly designed the active and passive beamforming to maximize the signal strength differences between adjacent users. Furthermore, Bazzi \emph{et al}. \cite{AhmadRIS} considered IRS beamforming design for joint detection and localization in multi-target scenarios, where the IRS beamforming was designed to maximize the strength of the IRS-collected target echo signals and minimize the strength of all other signals. }

\par \black{For facilitating sensing applications, some new IRS-assisted sensing architectures emerge, such as self-sensing IRS \cite{ShaoXiaodanTarget}, semi-passive IRS \cite{HuTWC,yuTSP,qian}, and target-mounted IRS \cite{WangPeilanTarget}. In \cite{ShaoXiaodanTarget}, the IRS controller was equipped with a radio frequency (RF) source for transmitting probing signals, and the IRS was installed with sensors to estimate the target direction by applying a multiple signal classification (MUSIC) algorithm. In \cite{HuTWC}, partial IRS elements can operate in sensing mode, and location information was estimated by combining the total least squares estimation of signal parameters via rotational invariance technique (TLS-ESPRIT) and the MUSIC algorithm. In \cite{WangPeilanTarget}, the IRS was mounted on the sensing target, and the target location and orientation were estimated as that of the IRS by leveraging the IRS beamforming and solving least-square problems.}

\par \black{For the new IRS-assisted sensing architectures, sensing information extraction can be solved through conventional signal processing techniques, due to their additional hardware structures (e.g., sensors and RF sources). For the passive IRS-aided sensing system, the transmitted signals undergo reconfigurable and cascaded echo channels due to the introduction of IRS, which changes the mapping relationship between sensing parameters and echo signals, leading to the intractability of extracting sensing information from the received echo signals and the sensitivity of sensing accuracy to the IRS passive beamforming. Nevertheless, the prevailing studies that use passive IRS as an anchor node to assist sensing tasks primarily concentrate on beamforming optimization under various scenarios and system architectures, while how to design effective signal processing approaches for extracting sensing information from IRS-reshaped echo signals is far from being well investigated.}
\par Motivated by the above issues, in this paper, we consider the non-line-of-sight (NLOS) target sensing problem in an IRS-assisted orthogonal frequency division multiplexing (OFDM) dual-function radar-communication (DFRC) system. To solve this problem, we propose an IRS-assisted NLOS target sensing framework, including sensing protocol, spatial signal processing for detection and direction estimation, as well as delay-Doppler (DD) signal processing for range/velocity estimation.  The main contributions of the article are as follows:
\begin{itemize}
    \item We develop an IRS-assisted NLOS target sensing protocol, where the whole coherent processing interval (CPI) is divided into the coarse-grained sensing (CGS) period and the fine-grained sensing (FGS) period. During the CGS period, the IRS conducts three-dimensional (3D) beam training to detect the presence of the target as well as to estimate its direction. With the IRS beam pointing at the target's direction estimated via CGS, the dual-function base station (DFBS) estimates the target range and velocity (R\&V) during the FGS period. By combining the estimated direction and range, the target location can be obtained.
    \item For the CGS period, we design a target detection and direction estimation scheme, based on the identification that the target presence and direction can be determined from the IRS-reshaped echo signals. \black{Specifically, first, invoking the sub-array partitioning and beam-broadening approaches, we devise a low-overhead hierarchical codebook for IRS 3D beamforming, which enables the IRS to adjust beam direction and width flexibly.} Then, we design a delay spectrum peak (DSP)-based hierarchical beam training (HBT) strategy to estimate the target direction, in which we propose a DSP detector for determining the target presence/absence within the beam scanning area. Finally, based on the linear interpolation technique, we propose a beam refinement (BR) method for further enhancing the direction estimation accuracy. 
    \item For the FGS period, we propose a target R\&V estimation scheme by extracting the DD information from the IRS-reshaped echo signals, where the beam training result obtained via the CGS is exploited to design the IRS beamforming of the FGS period for providing higher beamforming gain.
    \item Through both analytical and numerical results, we prove that with uniform power allocation among $N$ sub-carriers (SCs), using the DSP detector to process echo signals yields $N$ times signal processing gain than the power detector. Besides, the DSP-based HBT can reach a remarkable target detection success rate of $99.7\%$. \black{Moreover, after conducting the BR process, the target direction can be precisely estimated with $10^{-3}$-rad level accuracy.} In addition, by utilizing the IRS beam training result to assist FGS, the proposed R\&V estimation scheme for NLOS target sensing can achieve $10^{-6}$-m level range estimation accuracy and $10^{-5}$-m/s level velocity estimation accuracy, respectively.
\end{itemize}

\par The remainder of this paper is organized as follows. Section~\ref{section2} describes the 
system model of the IRS-assisted target sensing system. Section~\ref{section3} introduces the use of IRS 3D beam training for target detection and direction estimation, while Section~\ref{section4} presents the target R\&V estimation via DD estimation. Section~\ref{section5} extends the proposed system to the general multi-target case. The numerical results on the system performance are given in Section~\ref{section6}, and Section~\ref{section7} draw conclusions of this paper.

\emph{{Notations:}} Throughout this paper, the boldface upper/lower case represents matrices/vectors. $( \cdot ) ^{\mathrm{T}}$ and $( \cdot ) ^{\mathrm{H}}$ stand for transpose and Hermitian transpose, respectively. $\| \cdot \|$ and $\mathbb{E} \{ \cdot \} $ respectively denote the Euclidean norm and the expected value function. $\mathcal{C} \mathcal{N} ( 0,\sigma ^2 ) $ denotes the complex Gaussian distribution with mean $0$ and variance $\sigma ^2$. $\text{diag}( \cdot ) $ denotes the diagonal operation. Moreover, $\mathbb{P} \left( \cdot \right) $ is the probability of an event.
\section{System Model}\label{section2}
\vspace{-1mm}
\begin{figure}[htbp]
  \centering
  \includegraphics[width=3.2in]{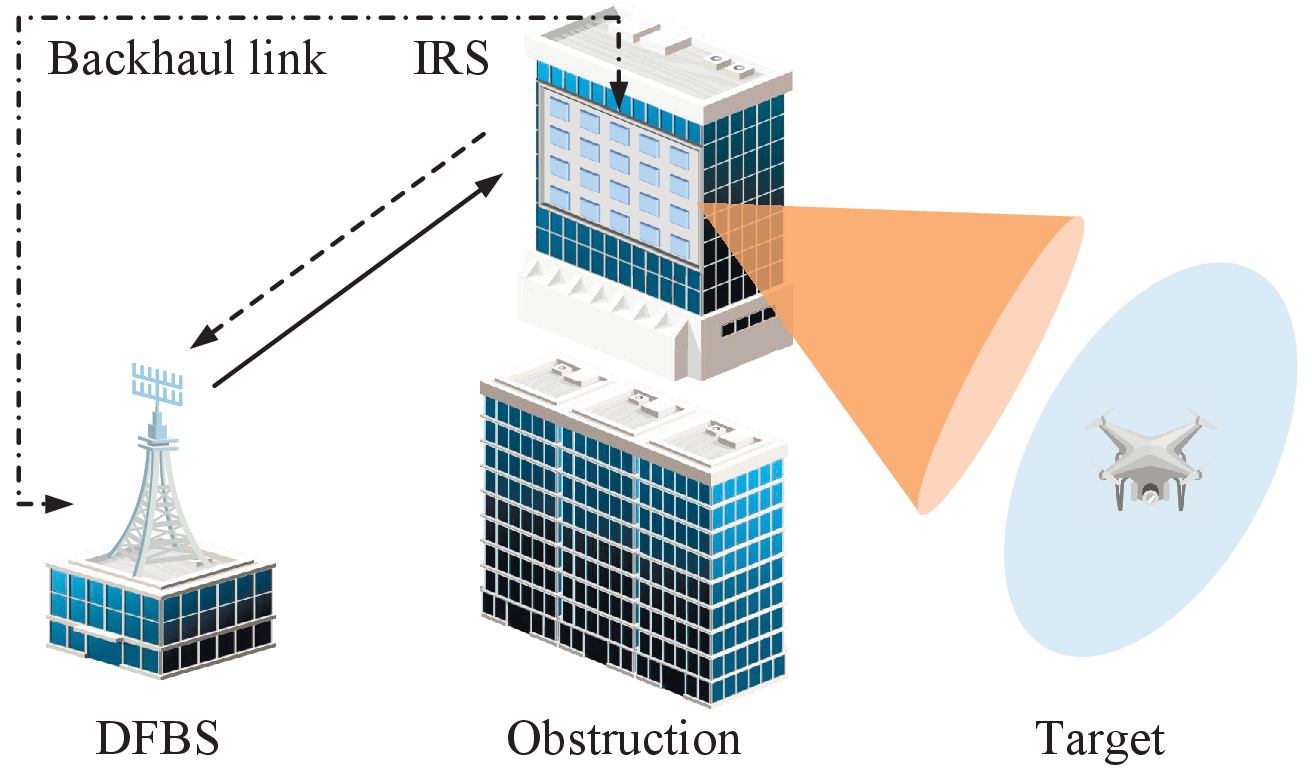}
  \caption{IRS-assisted target sensing system.}
  \label{system_model1}
\end{figure}
As shown in Fig.~\ref{system_model1}, this paper investigates an IRS-assisted NLOS target sensing system, where a DFBS with uniform linear array (ULA) consisting of $N_{\mathrm{B}}$ transmit/receive antennas, transmits the OFDM waveform and receives the echo signals reflected by an IRS with $M\!\!\times\!\! M$ elements. \footnote{The OFDM waveform is widely used for radar sensing \cite{han2013joint} and is a promising candidate for integrated sensing and communication\cite{10012421}, owing to its high spectral efficiency, robustness against multipath effect, and flexibility for resource allocation.} The ULA of the DFBS and the IRS are placed parallel to the $y$ axis and the $y$-$o$-$z$ plane, respectively. The overall bandwidth consists of $N$ equally-spaced SCs, where the frequency of the $n$-th ($n\!\in \!\mathcal{N} \!\triangleq \!\left\{ 0,\cdots ,N-1 \right\} $) SC is $f_n=f_c+n\varDelta f$, with $f_c$ and $\varDelta f$ respectively denote the carrier frequency and SC spacing.\footnote{We consider the scenario with constant wavelength $\lambda $ across the operational bandwidth\cite{9326394}.}
\subsection{Sensing protocol}
\begin{figure}[htbp]
  \centering
  \includegraphics[width=3.4in]{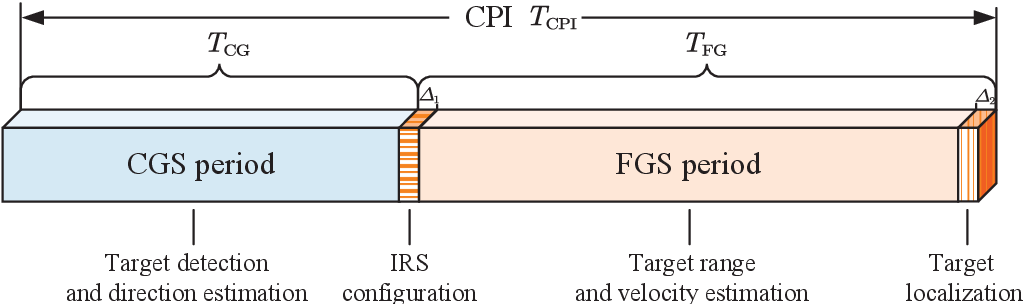}
  \caption{The sensing protocol during one CPI.}
  \label{protocol}
\end{figure}
As shown in Fig.~\ref{protocol}, we consider a CPI during which the target states (e.g., velocity and location) remain constant. One CPI is divided into two periods: the CGS period consisting of $T_{\mathrm{CG}}$ OFDM symbols, and the FGS period consisting of $T_{\mathrm{FG}}$ OFDM symbols. In the CGS period, the IRS conducts 3D beam training to detect the target presence as well as to estimate its direction. During the FGS period, the IRS is configured according to the beam training result of the CGS period. Then, the DFBS conducts target R\&V estimation by extracting DD information from the echo signals reshaped by the IRS. Finally, combining the target direction estimated during the CGS period and the target range estimated during the FGS period, the target location is determined. In addition, the first $\varDelta _1$ and last $\varDelta _2$ OFDM symbols are guard time for IRS configuration and location calculation, respectively, which are trivial compared to the whole CPI and will be ignored hereafter.
\vspace{-1mm}
\subsection{Channel Model}
\par The IRS is deployed in preferable locations that possess LoS links to both the DFBS and target. Therefore, the DFBS-IRS channel at the $n$-th ($n\in \mathcal{N}$) SC can be modeled by \cite{HurMil}
\begin{align}\label{HB2I}
\mathbf{G}_{\mathrm{B}2\mathrm{R},n}=\alpha _{\mathrm{B}2\mathrm{R},n}\mathbf{b}_{\mathrm{I}}\left( u_{\mathrm{B}2\mathrm{R}}^{\mathrm{A}},v_{\mathrm{B}2\mathrm{R}}^{\mathrm{A}} \right) \mathbf{a}_{\mathrm{B}}^{\mathrm{H}}\left( u_{\mathrm{B}2\mathrm{R}}^{\mathrm{D}} \right) ,
\end{align}
where $\alpha _{\mathrm{B}2\mathrm{R},n}$ is the complex channel gain for the DFBS-IRS link, which is described by \cite{ChenExploiting}
\begin{align}\label{alphaB2I}
\alpha _{\mathrm{B}2\mathrm{R},n}=a_{\mathrm{B}2\mathrm{R}}e^{\frac{-j2\pi f_nd_{\mathrm{B}2\mathrm{R}}}{c}},
\end{align}
where $d_{\mathrm{B}2\mathrm{R}}$ is the range from the DFBS to the IRS, $c$ is the lightspeed, $a_{\mathrm{B}2\mathrm{R}}$ denotes the path loss, which is given by  
\begin{align}
a_{\mathrm{B}2\mathrm{R}}=10^{-\frac{PL\left( d_0 \right)}{20}}\left( \frac{d_{\mathrm{B}2\mathrm{R}}}{d_0} \right) ^{-\frac{1}{2}\epsilon _{\mathrm{B}2\mathrm{R}}},
\end{align}
where $d_0$ is the reference distance, $\epsilon _{\mathrm{B}2\mathrm{R}}$ denotes the path loss exponent (PLE).
Besides, $\mathbf{a}_{\mathrm{B}}\left( \cdot \right) $/$\mathbf{b}_{\mathrm{I}}\left( \cdot \right)$ denote the DFBS/IRS array response vector. The normalized
 azimuth/elevation angles of arrival (AoA) (i.e., $u_{\mathrm{B}2\mathrm{R}}^{\mathrm{A}}$/$v_{\mathrm{B}2\mathrm{R}}^{\mathrm{A}}$) and the normalized angle of departure (AoD) (i.e., $u_{\mathrm{B}2\mathrm{R}}^{\mathrm{D}}$) are described by
\begin{align}
&u_{\mathrm{B}2\mathrm{R}}^{\mathrm{A}}=2 \frac{d_{\mathrm{R}}}{\lambda}\cos \left( \gamma _{\mathrm{B}2\mathrm{R}}^{\mathrm{A}} \right) \sin \left( \varphi _{\mathrm{B}2\mathrm{R}}^{\mathrm{A}} \right) ,
\\
&v_{\mathrm{B}2\mathrm{R}}^{\mathrm{A}}=2 \frac{d_{\mathrm{R}}}{\lambda}\sin \left( \gamma _{\mathrm{B}2\mathrm{R}}^{\mathrm{A}} \right) ,\quad u_{\mathrm{B}2\mathrm{R}}^{\mathrm{D}}=2 \frac{d_{\mathrm{B}}}{\lambda}\sin \left( \theta _{\mathrm{B}2\mathrm{R}}^{\mathrm{D}} \right) ,
\end{align}
where $d_{\mathrm{R}}$ and $d_{\mathrm{B}}$ respectively represent the element spacing of the IRS and the DFBS, $\theta _{\mathrm{B}2\mathrm{R}}^{\mathrm{D}}$, $\gamma _{\mathrm{B}2\mathrm{R}}^{\mathrm{A}}$, and $\varphi _{\mathrm{B}2\mathrm{R}}^{\mathrm{A}}$ are the AoD, elevation AoA, and azimuth AoA for the DFBS-IRS link, respectively.
\par Similarly, the IRS-target channel at the $n$-th ($n\in \mathcal{N}$) SC is modeled by
\begin{align}\label{hI2T}
\mathbf{g}_{\mathrm{R}2\mathrm{G},n}=\alpha _{\mathrm{R}2\mathrm{G},n}\mathbf{b}_{\mathrm{I}}^{\mathrm{H}}\left( u_{\mathrm{R}2\mathrm{G}}^{\mathrm{D}},v_{\mathrm{R}2\mathrm{G}}^{\mathrm{D}} \right),
\end{align}
where $\alpha _{\mathrm{R}2\mathrm{G},n}$ is the complex channel gain for the IRS-target link, which is given by 
\begin{align}
\alpha _{\mathrm{R}2\mathrm{G},n}=a_{\mathrm{R}2\mathrm{G}}\exp \left( \frac{-j2\pi f_nd_{\mathrm{R}2\mathrm{G}}}{c} \right)  ,
\end{align}
where $a_{\mathrm{R}2\mathrm{G}}$ represents the path loss, which is given by
\vspace{-2mm}
\begin{align}\label{aR2G}
a_{\mathrm{R}2\mathrm{G}}=10^{-\frac{PL\left( d_0 \right)}{20}}\left( \frac{d_{\mathrm{R}2\mathrm{G}}}{d_0} \right) ^{-\frac{1}{2}\epsilon _{\mathrm{R}2\mathrm{G}}},
\end{align}
with $\epsilon _{\mathrm{R}2\mathrm{G}}$ denotes the PLE. $d_{\mathrm{R}2\mathrm{G}}$ is the range from the IRS to the target. The normalized azimuth AoD and the normalized elevation AoD are respectively described by 
\begin{align}
&u_{\mathrm{R}2\mathrm{G}}^{\mathrm{D}}=2\frac{d_{\mathrm{R}}}{\lambda}\cos \left( \gamma _{\mathrm{R}2\mathrm{G}}^{\mathrm{D}} \right) \sin \left( \varphi _{\mathrm{R}2\mathrm{G}}^{\mathrm{D}} \right) ,\label{P43}
\\
&v_{\mathrm{R}2\mathrm{G}}^{\mathrm{D}}=2\frac{d_{\mathrm{R}}}{\lambda}\sin \left( \gamma _{\mathrm{R}2\mathrm{G}}^{\mathrm{D}} \right) ,\label{P44}
\end{align}
where $\gamma _{\mathrm{R}2\mathrm{G}}^{\mathrm{D}}$/$\varphi _{\mathrm{R}2\mathrm{G}}^{\mathrm{D}}$ denote the elevation/azimuth AoD for the IRS-target link.
\par For simplicity, we assume that $d_{\mathrm{R}}=d_{\mathrm{B}}=\lambda /2$. Hence, $\mathbf{a}_{\mathrm{B}}\left( u \right) $ and $\mathbf{b}_{\mathrm{I}}\left( u,v \right) $ are respectively modeled as
\begin{align}
\mathbf{a}_{\mathrm{B}}\left( u \right) &=\left[ 1,e^{j\pi u},\cdots ,e^{j\pi \left( N_{\mathrm{B}}-1 \right) u} \right] ^{\mathrm{T}},
\\
\mathbf{b}_{\mathrm{I}}\left( u,v \right) &=\left[ 1,e^{j\pi u},\cdots ,e^{j\pi \left( M-1 \right) u} \right] ^{\mathrm{T}}
\\
&\quad\otimes \left[ 1,e^{j\pi v},\cdots ,e^{j\pi \left( M-1 \right) v} \right] ^{\mathrm{T}}.\label{b}\notag
\end{align}
\subsection{Signal Model}
We pack the transmit symbol during one CPI into an $N\times T_{\mathrm{CPI}}$ matrix, which is described by
\begin{align}
\mathbf{S}=\small{\left[ \begin{matrix}
	s_{1,1}&		\cdots&		s_{1,T_{\mathrm{CPI}}}\\
	\vdots&		\ddots&		\vdots\\
	s_{N,1}&		\cdots&		s_{N,T_{\mathrm{CPI}}}\\
\end{matrix} \right]},
\end{align}
where $s_{n,l}$, satisfying $\mathbb{E} \{|s_{n,l}|^2\}=1$, is the $l$-th OFDM symbol on the $n$-th SC. Then, we describe the compact transmit precoding matrix by
\begin{align}
\mathbf{W}_{\mathrm{t}}=\mathrm{diag}\left( \mathbf{w}_{\mathrm{t},1},\cdots ,\mathbf{w}_{\mathrm{t},n} \right) \in \mathbb{C} ^{NN_{\mathrm{B}}\times N},
\end{align}
where $\mathbf{w}_{\mathrm{t},n}\in \mathbb{C} ^{N_{\mathrm{B}}\times 1}$ is the transmit precoding vector on the $n$-th SC, which satisfies $\left\| \mathbf{w}_{\mathrm{t},n} \right\| ^2=1$,
\par For signal transmission, $\mathbf{S}$ is precoded by $\mathbf{W}_{\mathrm{t}}$ in the frequency domain at first, which obtains the precoded symbol matrix
\begin{align}
\mathbf{X}=\mathbf{W}_{\mathrm{t}}\mathbf{S}=\small{\left[ \begin{matrix}
	\mathbf{x}_{1,1}&		\cdots&		\mathbf{x}_{1,T_{\mathrm{CPI}}}\\
	\vdots&		\ddots&		\vdots\\
	\mathbf{x}_{N,1}&		\cdots&		\mathbf{x}_{N,T_{\mathrm{CPI}}}\\
\end{matrix} \right]}\in \mathbb{C} ^{NN_{\mathrm{B}}\times T_{\mathrm{CPI}}},
\end{align}
where $\mathbf{x}_{n,l}=\mathbf{w}_{\mathrm{t},n}s_{n,l}\in \mathbb{C} ^{N_{\mathrm{B}}\times 1}$ denotes the $l$-th precoded symbol vector on the $n$-th SC. Subsequently, with an $N$-point inverse fast Fourier transform (IFFT), $\mathbf{X}$ is transformed into the time domain and inserted with a cyclic prefix (CP) whose length $T_{\mathrm{cp}}$ is longer than the multipath delay spread\cite{ZhengInte}. 
\par  After CP removal and conducting the $N$-point fast Fourier transform (FFT), the received baseband frequency-domain echo signals at the DFBS can be expressed as\footnote{The DFBS is full-duplex that transmits and receives simultaneously with self-interference cancellation \cite{SabharwalInBand}.}
\begin{align}
\tilde{\mathbf{y}}_{\mathrm{B},n,l}=&\sqrt{p_{\mathrm{T},n}}\alpha _{\mathrm{G}}\bar{\mathbf{G}}_n\left( \boldsymbol{\xi }\left( l \right) \right) \mathbf{x}_{n,l}e^{j2\pi lT_{\mathrm{O}}f_{\mathrm{D}}}+\mathbf{n}_{n,l},\notag
\\
&\quad n\in \mathcal{N} ,l\in \mathcal{L} \triangleq \left\{ 0,\cdots ,\left( T_{\mathrm{CPI}}-1 \right) \right\} ,
\end{align}
where $p_{\mathrm{T},n}$ is the allocated power on the $n$-th SC, $\alpha _{\mathrm{G}}$ is the target radar cross section (RCS) with $\mathbb{E} \{ \left| \alpha _{\mathrm{G}} \right|^2 \} =\zeta _{\mathrm{G}}^{2}$, $T_{\mathrm{O}}=1/\varDelta f+T_{\mathrm{cp}}$ represents the OFDM symbol duration, $f_{\mathrm{D}}=2\mathrm{v}_{\mathrm{G}}/\lambda $ denotes the Doppler shift with $\mathrm{v}_{\mathrm{G}}$ being the target velocity, respectively\cite{BraunMaximum}. $\mathbf{n}_{n,l}\sim \mathcal{C} \mathcal{N} \left( 0,\sigma _{0}^{2}\mathbf{I}_{N_{\mathrm{B}}} \right) $ denotes the additive white Gaussian noise (AWGN). In addition, the effective target echo channel at the $n$-th SC is defined as
\begin{align}
&\bar{\mathbf{G}}_n\left( \boldsymbol{\xi }\left( l \right) \right) =
\\
&\mathbf{G}_{\mathrm{B}2\mathrm{R},n}^{\mathrm{T}}\mathrm{diag}\left( \boldsymbol{\xi }\left( l \right) \right) \mathbf{g}_{\mathrm{R}2\mathrm{G},n}^{\mathrm{T}}\mathbf{g}_{\mathrm{R}2\mathrm{G},n}\mathrm{diag}\left( \boldsymbol{\xi }\left( l \right) \right) \mathbf{G}_{\mathrm{B}2\mathrm{R},n},\notag
\end{align}
where $\boldsymbol{\xi }\left( l \right) =\left[ e^{j\vartheta _1\left( l \right)},\cdots ,e^{j\vartheta _m\left( l \right)},\cdots ,e^{j\vartheta _{M^2}\left( l \right)} \right] ^{\mathrm{T}}$ denotes the IRS phase shift beam at the $l$-th OFDM symbol.
\par Once the DFBS receives the echo signals, the signals are beamformed using the receive combiner, which obtains the received signal matrix $\mathbf{Y}_{\mathrm{B}}$, whose $(n,l)$-th element is 
\begin{align}\label{yBS}
&\left[ \mathbf{Y}_{\mathrm{B}} \right] _{n,l}\triangleq y_{\mathrm{B},n,l}=\mathbf{w}_{\mathrm{r},n}^{\mathrm{T}}\tilde{\mathbf{y}}_{\mathrm{B},n,l}
\\
&\qquad=\sqrt{p_{\mathrm{T},n}}\alpha _{\mathrm{G}}\mathbf{w}_{\mathrm{r},n}^{\mathrm{T}}\bar{\mathbf{G}}_n\left( \boldsymbol{\xi }\left( l \right) \right) \mathbf{x}_{n,l}e^{j2\pi lT_{\mathrm{O}}f_{\mathrm{D}}}+\mathbf{w}_{\mathrm{r},n}^{\mathrm{T}}\mathbf{n}_{n,l},\notag
\end{align}
where $\mathbf{w}_{\mathrm{r},n}\in \mathbb{C} ^{N_{\mathrm{B}}\times 1}$, satisfying $\left\| \mathbf{w}_{\mathrm{r},n} \right\| ^2=1$, is the receive combining vector on the $n$-th SC.

\section{Target Detection and Direction Estimation via IRS 3D Beam Training}\label{section3}
\par During the CGS period, an IRS 3D beam training scheme is proposed to detect the presence of the target and estimate its direction. Specifically, first, we design a hierarchical codebook for IRS 3D beamforming. Then, based on the hierarchical codebook, a low-overhead 3D HBT scheme is proposed for estimating the target direction, where a DSP detector is devised to determine the presence of the target within the current beam scanning area. Finally, we propose a BR method for further enhancing the direction estimation accuracy. 
\subsection{Problem Statement}
\par First, for the beamforming design at the DFBS, we denote $A_{\mathrm{B}}\left( \mathbf{w},u_{\mathrm{B}} \right) \triangleq \left| \mathbf{a}_{\mathrm{B}}^{\mathrm{H}}\left( u_{\mathrm{B}} \right) \mathbf{w} \right|$ as the beam gain of $ \mathbf{w}$ towards the spatial direction $u_{\mathrm{B}}\in \left[ -1,1 \right] $, where $ \mathbf{w}$ denotes either $\mathbf{w}_{\mathrm{t},n}$ or $\mathbf{w}_{\mathrm{r},n}$. For facilitating the IRS-assisted sensing task of the NLOS target, the DFBS beamforming should be designed to maximize beam gain towards the IRS. Hence, $\mathbf{w}_{\mathrm{t},n}$ and $\mathbf{w}_{\mathrm{r},n}$ can be designed as
\begin{align}\label{BSopt}
\mathbf{w}_{\mathrm{t},n}^{\mathrm{opt}}=\mathbf{w}_{\mathrm{r},n}^{\mathrm{opt}}&=\mathop {\mathrm{arg}\max} \limits_{\mathbf{w}}A_{\mathrm{L}}\left( \mathbf{w},u_{\mathrm{B}2\mathrm{R}}^{\mathrm{D}} \right) \\
&=\mathbf{w}_{\mathrm{B}}^{\mathrm{opt}}\triangleq \frac{1}{\sqrt{N_{\mathrm{B}}}}\mathbf{a}_{\mathrm{B}}\left( u_{\mathrm{B}2\mathrm{R}}^{\mathrm{D}} \right) , n\in \mathcal{N}.\notag
\end{align}
\par Then, for the beamforming design at the IRS, we denote $A_{\mathrm{R}}\left( \boldsymbol{\xi },u_{\mathrm{R}},v_{\mathrm{R}} \right) $ as the beam gain of $\boldsymbol{\xi }$ towards the spatial direction $\left( u_{\mathrm{R}},v_{\mathrm{R}} \right) $ ($u_{\mathrm{R}},v_{\mathrm{R}}\in \left[ -1,1 \right] $), which is described by
\begin{align}
A_{\mathrm{R}}\!\left( \boldsymbol{\xi },\!u_{\mathrm{R}},\!v_{\mathrm{R}} \right) \!\triangleq \!\left| \mathbf{b}_{\mathrm{I}}^{\mathrm{H}}\!\left( u_{\mathrm{R}},\!v_{\mathrm{R}} \right) \mathrm{diag}\!\left( \boldsymbol{\xi } \right) \mathbf{b}_{\mathrm{I}}\!\left( u_{\mathrm{B}2\mathrm{R}}^{\mathrm{A}},\!v_{\mathrm{B}2\mathrm{R}}^{\mathrm{A}} \right)  \right|.
\end{align}
The phase shift beam which maximizes the beam gain along the direction $\left( u_{\mathrm{R}},v_{\mathrm{R}} \right) $ satisfies
\begin{subequations} 
\begin{align}
\text {(P1)}: \max_{\boldsymbol{\xi }}  \,\,&A_{\mathrm{R}}\!\left( \boldsymbol{\xi },\!u_{\mathrm{R}},\!v_{\mathrm{R}} \right) ,\\
\text{s.t.}\quad\,&\left| \left[ \boldsymbol{\xi } \right] _m \right|=1,m\in \left\{ 1,\cdots M \right\} ,
\end{align}
\end{subequations}
which yields
\begin{align}\label{anglerelation}
\boldsymbol{\xi }^{\mathrm{opt}}\left( u_{\mathrm{R}},v_{\mathrm{R}} \right) &=\mathrm{diag}\left( \mathbf{b}_{\mathrm{I}}^{*}\left( u_{\mathrm{B}2\mathrm{R}}^{\mathrm{A}},v_{\mathrm{B}2\mathrm{R}}^{\mathrm{A}} \right) \right) \mathbf{b}_{\mathrm{I}}\left( u_{\mathrm{R}},v_{\mathrm{R}} \right)  \notag
\\
&=\mathbf{b}_{\mathrm{I}}^{*}\left( u_{\mathrm{B}2\mathrm{R}}^{\mathrm{A}},v_{\mathrm{B}2\mathrm{R}}^{\mathrm{A}} \right) \circ \mathbf{b}_{\mathrm{I}}\left( u_{\mathrm{R}},v_{\mathrm{R}} \right) \notag
\\
&=\mathbf{b}_{\mathrm{I}}\left( u_{\mathrm{R}}-u_{\mathrm{B}2\mathrm{R}}^{\mathrm{A}},v_{\mathrm{R}}-v_{\mathrm{B}2\mathrm{R}}^{\mathrm{A}} \right) .
\end{align}
Hence, a straightforward approach for target detection and direction estimation is using the exhaustive beam searching (EBS) scheme to search the entire angle range \cite{WangJoint}. However, considering the massive number of IRS elements, the EBS scheme is impractical since it requires the overhead of at least $M^2$ OFDM symbols. Responding to this, we design a low-overhead yet efficient hierarchical codebook for IRS 3D beam training in the following.
\subsection{Hierarchical codebook design}
\vspace{-2mm}
\begin{figure}[ht]
\vspace{-2mm}
  \centering
  \subfigure[Normalized azimuth domain.]  
  {
  \label{u}
  \includegraphics[width=2.5in]{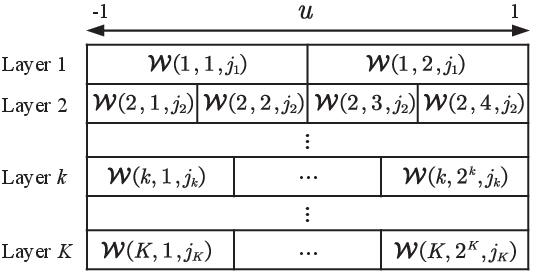}
  }
  \subfigure[Normalized elevation domain.]
  {
  \label{v}
  \includegraphics[width=2.5in]{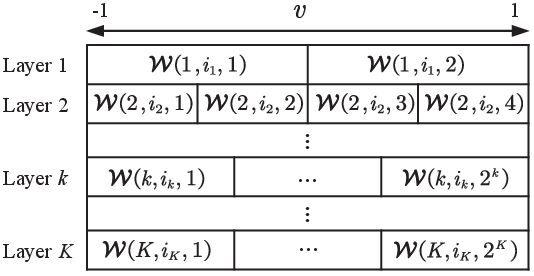}
  }
  \caption{Beam coverage of the Hierarchical codebook.}
\end{figure}
Fig.~\ref{u}/Fig.~\ref{v} illustrates the spatial coverage of the hierarchical codebook in the normalized azimuth/elevation direction, where the $\left( i_k,j_k \right)$-th codeword on the $k$-th layer is denoted by $\bm{\mathcal{W}}\left( k,i_k,j_k \right) \in \mathbb{C} ^{M^2\times 1}$,
with the azimuth/elevation index $i_k/j_k$ belonging to the set $\mathcal{F} _k=\left\{ 1,\cdots ,2^k \right\} $. Specifically, the hierarchical codebook has $K= \log _2M$ layers, with $M$ being a dyadic integer. For the $k$-th layer, there are totally $4^k$ codewords, the $\left( i_k,j_k \right)$-th codeword points to the normalized spacial direction $\left( -1+\frac{2i_k-1}{2^k},-1+\frac{2j_k-1}{2^k} \right)$ with beam width $2/2^k$, and the union of $4^k$ codewords covers the entire angle range  $\Omega \triangleq \left\{ \left( u,v \right) |u,v\in \left[ -1,1 \right] \right\} $. Below we provide the details of the hierarchical codebook design. 
\par First, for the codewords in the $K$-th layer that form narrow beams with beam width $2/M$, we design them to follow the form of $\mathbf{b}_{\mathrm{I}}$ and uniformly distribute in $\Omega $
\begin{align}
\bm{\mathcal{W}}\left( K,i_K,j_K \right) =\mathbf{b}_{\mathrm{I}}\left(\frac{2i_K-1}{M}-1,\frac{2j_K\!-\!1}{M}-1 \right) .
\end{align}
\par Then, for the codewords in the $k$-th ($k\in \left\{ 1,\cdots ,K-1 \right\} $) layer that form wide beams with beam width $2/2^k$, our basic idea is to broaden the beam layer-by-layer. Specifically, it is well known that concentrating the power of an $N$-element ULA in a specific direction yields a beam with beam width $2/N$. If the ULA is partitioned into $S$ sub-arrays that generate $S$ sub-beams pointing at sufficiently-spaced directions, its beam width can be broadened by $S^2$ times, as depicted in Fig.~\ref{broaden}. Motivated by this, we propose to design the IRS codewords by combining the sub-array partitioning and beam-broadening approaches. 
\begin{figure}[htbp]
  \centering
   \includegraphics[width=3.2in]{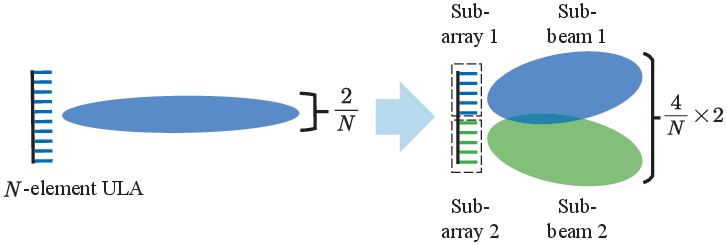}
  \caption{Beam-broadening for the case of $S=2$.}
  \label{broaden}
\end{figure}
\par Let $l\triangleq K-k$. For the $k$-th layer, if $l$ is even, we first partition the IRS into $S_{\mathrm{e},l}^2$ sub-arrays
each with $M_{\mathrm{e},l}\times M_{\mathrm{e},l}$ elements, where $S_{\mathrm{e},l}\triangleq 2^{l/2}$ and $M_{\mathrm{e},l}\triangleq M/S_{\mathrm{e},l}$. Then, we recast the codeword as
\begin{align}
\bm{\mathcal{W}}\left( k,i_k,j_k \right) =\bm{\mathcal{V}}_{\mathrm{e},u}\left( k,i_k \right) \otimes \bm{\mathcal{V}}_{\mathrm{e},v}\left( k,j_k \right) ,
\end{align}
where $\bm{\mathcal{V}}_{\mathrm{e},u}\left( k,i_k \right) 
\in \mathbb{C} ^{M\times 1}$ and $\bm{\mathcal{V}}_{\mathrm{e},v}\left( k,j_k \right) \in \mathbb{C} ^{M\times 1}$ respectively denote the azimuth and elevation beamforming vectors, which can be designed separately with the same procedure due to the same structure and decoupled relationship.
\par For the sake of generality, we focus on designing the azimuth beamforming vectors $\bm{\mathcal{V}}_{\mathrm{e},u}\left( k,i_k \right) $. Specifically, based on the sub-array partitioning approach, we transform the design of $\bm{\mathcal{V}}_{\mathrm{e},u}\left( k,i_k \right)$ into the design of an $M$-element ULA, which consists of $S_{\mathrm{e},l}$ sub-ULAs each with $M_{\mathrm{e},l}$ elements. Since an $M_{\mathrm{e},l}$-element sub-ULA can form a sub-beam with beam width 
\begin{align}
\frac{2}{M_{\mathrm{e},l}}=\frac{2S_{\mathrm{e},l}}{M}=\frac{2}{2^k2^{\frac{K-k}{2}}}=\frac{2}{2^kS_{\mathrm{e},l}}.
\end{align}
Based on the beam-broadening approach, a wide beam centered at $\psi \left( k,i_k \right) \triangleq -1+\frac{2i_k-1}{2^k}$ with beam width $2/2^k$ can be obtained by letting the $S_{\mathrm{e},l}$ sub-ULAs form $S_{\mathrm{e},l}$ sub-beams with $2/2^kS_{\mathrm{e},l}$ sub-beam spacing. Therefore, the beamforming of the $s$-th ($s=1,\cdots ,S_{\mathrm{e},l}$) sub-ULA is designed as
\begin{align}\label{sub-ULA}
\bm{\mathcal{V}}_{\mathrm{e},u}^{s}\left( k,i_k \right) =\mathbf{a}\left( M_{\mathrm{e},l},\frac{2s-1}{2^kS_{\mathrm{e},l}}+\psi \left( k,i_k \right) -\frac{1}{2^k} \right),
\end{align}
where
\begin{align}
\mathbf{a}\left( N_a,\psi \right) \triangleq \left[ 1,\cdots ,e^{j\pi \left( n-1 \right) \psi},\cdots ,e^{j\pi \left( N_a-1 \right) \psi} \right] ^{\mathrm{T}}.
\end{align}
As such, $\bm{\mathcal{V}}_{\mathrm{e},u}\left( k,i_k \right) $ are obtained as
\begin{align}
&\bm{\mathcal{V}}_{\mathrm{e},u}\left( k,i_k \right) =\left[ e^{-j\zeta _{\mathrm{e},l}}\bm{\mathcal{V}}_{\mathrm{e},u}^{1}\left( k,i_k \right) ^{\mathrm{T}},\cdots \right. 
\\
&\quad \left. ,e^{-js\zeta _{\mathrm{e},l}}\bm{\mathcal{V}}_{\mathrm{e},u}^{s}\left( k,i_k \right) ^{\mathrm{T}},\cdots ,e^{-jS_{\mathrm{e},l}\zeta _{\mathrm{e},l}}\bm{\mathcal{V}}_{\mathrm{e},u}^{S_{\mathrm{e},l}}\left( k,i_k \right) ^{\mathrm{T}} \right] ^{\mathrm{T}},\notag
\end{align}
where $\zeta _{\mathrm{e},l}=\frac{M_{\mathrm{e},l}-1}{M_{\mathrm{e},l}}\pi $, $e^{-js\zeta _{\mathrm{e},l}}$ is designed for reducing the fluctuation of the beam\cite{XiaoHiera}, which is called beam compensation (BC) coefficient. 
\par Following the same procedure, we designed $\bm{\mathcal{V}}_{\mathrm{e},v}\left( k,j_k \right)$, and the codeword $\bm{\mathcal{W}}\left( k,i_k,j_k \right) $ is finally obtained as
\begin{align}
\bm{\mathcal{W}}\left( k,i_k,j_k \right) =\bm{\mathcal{V}}_{\mathrm{e},u}\left( k,i_k \right) \otimes \bm{\mathcal{V}}_{\mathrm{e},u}\left( k,j_k \right) .
\end{align}

\par For the $k$-th layer where $l=K-k$ is odd, we first equally divide the IRS into $4$ groups, in which the first group forms a wide beam while others point to the edge of $\Omega $ to avoid interfering with the first group. Then, we partition the first group into $S_{\mathrm{o},l}^2$ sub-arrays each with $M_{\mathrm{o},l}\times M_{\mathrm{o},l}$ IRS elements, where
$S_{\mathrm{o},l}\triangleq 2^{\frac{l-1}{2}}$ and $M_{\mathrm{o},l}\triangleq M/2S_{\mathrm{o},l}$. Next, we recast the codeword as
\begin{align}
\bm{\mathcal{W}}\left( k,i_k,j_k \right) =\left[ \bm{\mathcal{W}}_{\mathrm{1}}^{\mathrm{T}}\left( k,i_k,j_k \right) ,\bm{\mathcal{E}}^{\mathrm{T}},\bm{\mathcal{E}}^{\mathrm{T}},\bm{\mathcal{E}}^{\mathrm{T}} \right] ^{\mathrm{T}},
\end{align}
where $\bm{\mathcal{E}}$ is the beamforming vector of the $i$-th ($i=2,3,4$) group, which is designed to point towards the edge of $\Omega$, i.e.,
\begin{align}
\bm{\mathcal{E}}=\mathbf{a}\left( \frac{M}{2},1 \right) \otimes \mathbf{a}\left( \frac{M}{2},1 \right) ,
\end{align}
and $\bm{\mathcal{W}}_{\mathrm{1}}$ denotes the beamforming vector of the first group, for designing which, we recast it as
\begin{align}
\bm{\mathcal{W}}_{\mathrm{1}}\left( k,i_k,j_k \right) =\bm{\mathcal{V}}_{\mathrm{o},u}\left( k,i_k \right) \otimes \bm{\mathcal{V}}_{\mathrm{o},v}\left( k,j_k \right),
\end{align}
where $\bm{\mathcal{V}}_{\mathrm{o},u}( k,i_k ) 
\in \mathbb{C} ^{M/2\times 1}$ and $\bm{\mathcal{V}}_{\mathrm{o},v}( k,j_k ) \in \mathbb{C} ^{M/2\times 1}$ respectively denote the azimuth and elevation beamforming vectors, which can also be designed separately with the same procedure. 
\par For the sake of generality, we focus on designing the
azimuth beamforming vectors $\bm{\mathcal{V}}_{\mathrm{o},u}\left( k,i_k \right) $, and transform it into the beamforming design of an $M/2$-element ULA, which consists of $S_{\mathrm{o},l}$ sub-ULAs each with $M_{\mathrm{o},l}$ elements. Since an $M_{\mathrm{o},l}$-element sub-array can form a sub-beam with the beam width
\begin{align}
\frac{2}{M_{\mathrm{o},l}}=\frac{4S_{\mathrm{o},l}}{M}=\frac{2}{2^{\frac{K+k-1}{2}}}=\frac{2}{2^kS_{\mathrm{o},l}}.
\end{align}
Therefore, similar to designing $\bm{\mathcal{V}}_{\mathrm{e},u}\left( k,i_k \right) $, we design the beamforming of the $s$-th ($s=1,\cdots ,S_{\mathrm{o},l}$) sub-ULA as
\begin{align}
\bm{\mathcal{V}}_{\mathrm{o},u}^{s}\left( k,i_k \right) =\mathbf{a}\left( M_{\mathrm{o},l},\frac{2s-1}{2^kS_{\mathrm{o},l}}+\psi \left( k,i_k \right) -\frac{1}{2^k} \right),
\end{align}
and design $\bm{\mathcal{V}}_{\mathrm{o},u}\left( k,i_k \right) $ as
\begin{align}
&\bm{\mathcal{V}}_{\mathrm{o},u}\left( k,i_k \right) =\left[ e^{-j\zeta _{\mathrm{o},l}}\bm{\mathcal{V}}_{\mathrm{o},u}^{1}\left( k,i_k \right) ^{\mathrm{T}},\cdots \right.  
\\
&\quad \left. ,e^{-js\zeta _{\mathrm{o},l}}\bm{\mathcal{V}}_{\mathrm{o},u}^{s}\left( k,i_k \right) ^{\mathrm{T}},\cdots ,e^{-jS_{\mathrm{o},l}\zeta _{\mathrm{o},l}}\bm{\mathcal{V}}_{\mathrm{o},u}^{S_{\mathrm{o},l}}\left( k,i_k \right) ^{\mathrm{T}} \right] ^{\mathrm{T}},\notag
\end{align}
where the $e^{-js\zeta _{\mathrm{o},l}}$ with $\zeta _{\mathrm{o},l}=\frac{M_{\mathrm{o},l}-1}{M_{\mathrm{o},l}}\pi $ being the BC coefficient.
\par Finally, the codeword $\bm{\mathcal{W}}\left( k,i_k,j_k \right) $ is obtained as
\begin{align}
&\bm{\mathcal{W}}\left( k,i_k,j_k \right) =\left[ \bm{\mathcal{W}}_{\mathrm{1}}\left( k,i_k,j_k \right) ^{\mathrm{T}},\bm{\mathcal{E}}^{\mathrm{T}},\bm{\mathcal{E}}^{\mathrm{T}},\bm{\mathcal{E}}^{\mathrm{T}} \right] ^{\mathrm{T}}\notag
\\
&\quad=\!\left[ \left( \bm{\mathcal{V}}_{\mathrm{o},u}\left( k,i_k \right) \otimes \bm{\mathcal{V}}_{\mathrm{o},v}\left( k,j_k \right) \right) ^{\mathrm{T}},\bm{\mathcal{E}}^{\mathrm{T}},\bm{\mathcal{E}}^{\mathrm{T}},\bm{\mathcal{E}}^{\mathrm{T}} \right] ^{\mathrm{T}},
\end{align}
where $\bm{\mathcal{V}}_{\mathrm{o},v}\left( k,j_k \right)$ is designed similar to $\bm{\mathcal{V}}_{\mathrm{o},u}\left( k,i_k \right)$.
\par Based on the above procedures, the designed hierarchical codebook can flexibly adjust the direction and width of the IRS beam.

\subsection{3D HBT based on a DSP detector}
\begin{figure}[htbp]
  \centering
  \includegraphics[width=3.4in]{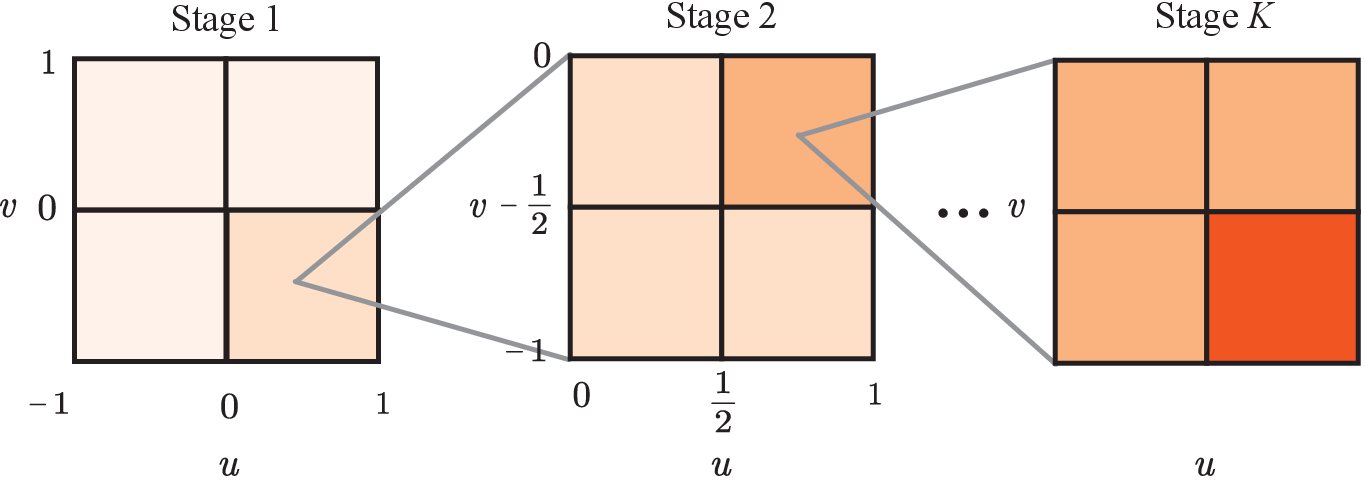}
  \caption{Illustration of 3D HBT.}
  \label{Beam search}
\end{figure}
\par As illustrated in Fig.~\ref{Beam search}, the 3D HBT contains $K$ stages, during each of which $4$ IRS beams are trained. Specifically, in stage $1$, the IRS forms $4$ beams in sequence to scan the entire angle range, based on the $4$ codewords of the first layer, where the duration of each beam is $T_{\mathrm{O}}$. Based on the received echo signals corresponding to the $4$ beams, the DFBS determines the presence of the target within the $4$ beam coverages based on a test statistic. In particular, the target is judged to be present in the coverage area of the beam that maximizes the test statistic. In stage $k$ ($k>1$), the beam with the largest test statistic in stage $(k-1)$ is selected, and its coverage will be equally divided into $4$ parts. Then, the IRS forms $4$ training beams in sequence to cover these $4$ parts based on the $4$ codewords of the codebook's $k$-th layer, and the DFBS detects the target as before. 

\par Below we will provide the details for selecting the test statistic and detecting the target, respectively. 
\subsubsection{Test statistic selection}
\par For target detection, the received signal strength (RSS) is a widely adopted test statistic, which is described by
\begin{align}\label{RSS}
\mathcal{R} _l=\sum_{n=0}^{N-1}{\left| y_{\mathrm{B},n,l} \right|^2}.
\end{align}
\par \black{However, for the HBT scheme that starts at wide beams with low beamforming gain, the echo signals are overwhelmed by noise in the earlier training stage. To cope with this issue, we propose a new test statistic, i.e., the DSP, for exploiting the correlation among multiple SCs and achieving effective noise suppression, which is defined as}
\begin{align}
\mathcal{P} _l=\underset{n_{\mathrm{Q}}\in \mathcal{N} _{\mathrm{Q}}}{\max}\frac{1}{N}\left| \sum_{n=0}^{N-1}{\frac{y_{\mathrm{B},n,l}}{s_{n,l}}e^{j2\pi n\varDelta f\tau _{n_{\mathrm{Q}}}}} \right|^2,
\end{align}
where 
\begin{align}
\tau _{n_{\mathrm{Q}}}=\frac{n_{\mathrm{Q}}}{N_{\mathrm{Q}}}\frac{1}{\varDelta f},n_{\mathrm{Q}}\in \mathcal{N} _{\mathrm{Q}}=\left\{ 0,\cdots ,\left( N_{\mathrm{Q}}-1 \right) \right\} .
\end{align}
\begin{proposition}
With uniform power allocation among all SCs, using the proposed DSP test statistic for echo signal processing yields $N$ times signal processing gain compared with the RSS test statistic.
\end{proposition}

\begin{proof}
Based on (\ref{HB2I}), (\ref{hI2T}), and (\ref{yBS}), the DFBS received echo signals can be expressed as 
\begin{align}\label{ybsreformulate}
&y_{\mathrm{B},n,l}\!=\!\sqrt{p_{\mathrm{T},n}}\alpha _{\mathrm{G}}\mathbf{w}_{\mathrm{B}}^{\mathrm{opt},\mathrm{T}}\bar{\mathbf{G}}_n\left( \boldsymbol{\xi }\left( l \right) \right) \mathbf{w}_{\mathrm{B}}^{\mathrm{opt}}s_{n,l}e^{j2\pi lT_{\mathrm{O}}f_{\mathrm{D}}}\!+\!\bar{n}_{n,l}\notag
\\
&\quad=\sqrt{p_{\mathrm{T},n}}\alpha _{\mathrm{CG},l}e^{j\beta _{\mathrm{CG},l}}e^{-j2\pi n\varDelta f\left( \tau _0+\tau \right)}s_{n,l}+\bar{n}_{n,l},
\end{align}
where 
\begin{align}
&\alpha _{\mathrm{CG},l}\triangleq N_{\mathrm{B}}\alpha _{\mathrm{G}}a_{\mathrm{R}2\mathrm{G}}^{2}a_{\mathrm{B}2\mathrm{R}}^{2}\left| h_{\boldsymbol{\xi },l}^{2} \right|,\label{alphaCG}
\\
&\beta _{\mathrm{CG},l}\triangleq \mathrm{angle}\left( h_{\boldsymbol{\xi },l}^{2} \right) +2\pi \left( lT_{\mathrm{O}}f_{\mathrm{D}}\!-\!f_c\left( \tau _0+\tau  \right) \right) ,
\\
&\tau _0\triangleq 2d_{\mathrm{B}2\mathrm{R}}/c,\quad\tau \triangleq 2d_{\mathrm{R}2\mathrm{G}}/c,\quad\bar{n}_{n,l}\triangleq \mathbf{w}_{\mathrm{B}}^{\mathrm{opt},\mathrm{T}}\mathbf{n}_{n,l},
\end{align}
where $h_{\boldsymbol{\xi },l}=\mathbf{b}_{\mathrm{I}}^{\mathrm{H}}\left( u_{\mathrm{R}2\mathrm{G}}^{\mathrm{D}},v_{\mathrm{R}2\mathrm{G}}^{\mathrm{D}} \right) \mathrm{diag}\left( \boldsymbol{\xi }\left( l \right) \right) \mathbf{b}_{\mathrm{I}}\left( u_{\mathrm{B}2\mathrm{R}}^{\mathrm{A}},v_{\mathrm{B}2\mathrm{R}}^{\mathrm{A}} \right)  $. Then, let $\mathbf{f}_l$ denote an auxiliary vector, whose $n$-th element is given by
\begin{align}\label{fl}
\left[ \mathbf{f}_l \right] _n&=e^{j2\pi n\varDelta f\tau _0}\frac{y_{\mathrm{B},n,l}}{s_{n,l}}\!
\\
&=\sqrt{p_{\mathrm{T},n}}\alpha _{\mathrm{CG},l}e^{j\beta _{\mathrm{CG},l}}e^{-j2\pi n\varDelta f\tau }+e^{j2\pi n\varDelta f\tau _0}\frac{\bar{n}_{n,l}}{s_{n,l}},\notag
\end{align}
Next, define the quantified delay as
\begin{align}
\tau _{n_{\mathrm{Q}}}\triangleq \frac{n_{\mathrm{Q}}}{N_{\mathrm{Q}}}\frac{1}{\varDelta f},n_{\mathrm{Q}}\in \mathcal{N} _{\mathrm{Q}}\triangleq \left\{ 0,\cdots ,\left( N_{\mathrm{Q}}-1 \right) \right\} ,
\end{align}
and calculate the delay spectrum $\boldsymbol{\varGamma }_l$ as
\begin{align}\label{delayspectrum}
\left[ \boldsymbol{\varGamma }_l \right] _{n_{\mathrm{Q}}}=\frac{1}{N}\left| \sum_{n=0}^{N-1}{\left[ \mathbf{f}_l \right] _ne^{j2\pi n\varDelta f\tau _{n_{\mathrm{Q}}}}} \right|^2.
\end{align}
Note that, when $N_{\mathrm{Q}}$ is enough large, there exists a $\tau _{n_{\mathrm{Q}}}$ that satisfies
\begin{align}
\tau _{\hat{n}_{\mathrm{D}}}=\underset{\tau _{n_{\mathrm{Q}}}}{\mathrm{arg}}\,\,\,\underset{n_{\mathrm{Q}}\in \mathcal{N} _{\mathrm{Q}}}{\min}\left| \tau  -\tau _{n_{\mathrm{Q}}} \right|\approx \tau  ,
\end{align}
which makes
\begin{align}
e^{j2\pi n\varDelta f\tau _{\hat{n}_{\mathrm{Q}}}}e^{-j2\pi n\varDelta f\tau }\approx 1,
\end{align}
and the $\boldsymbol{\varGamma }_l$ meet its peak value at $\left[ \boldsymbol{\varGamma }_l \right] _{\hat{n}_{\mathrm{Q}}}$. Hence, we express the DSP as (\ref{DSP_ori}) at the top of the next page.
\begin{figure*}[t]
\begin{equation}
 \begin{gathered}
  \begin{aligned}
    \mathcal{P} _l=\left[ \boldsymbol{\varGamma }_l \right] _{\hat{n}_{\mathrm{Q}}}\approx \left| \frac{1}{\sqrt{N}}\sum_{n=0}^{N-1}{\sqrt{p_{\mathrm{T},n}}\alpha _{\mathrm{CG},l}e^{j\beta _{\mathrm{CG},l}}}+\frac{1}{\sqrt{N}}\sum_{n=0}^{N-1}{\frac{\bar{n}_{n,l}}{s_{n,l}}e^{j2\pi n\varDelta f\left( \tau _{\hat{n}_{\mathrm{Q}}}+\tau _0 \right)}} \right|^2.\label{DSP_ori}
 \end{aligned}
 \end{gathered}
\end{equation}
\begin{equation}
 \begin{gathered}
  \begin{aligned}
    \mathcal{R} _l=\sum_{n=0}^{N-1}{\left| \sqrt{p_{\mathrm{T},n}}\alpha _{\mathrm{CG},l}e^{j\left( \beta _{\mathrm{CG},l}-2\pi n\varDelta f\left( \tau _0+\tau  \right) \right)}s_{n,l}+\bar{n}_{n,l} \right|^2}.\label{RSS2}
 \end{aligned}
 \end{gathered}
\end{equation}
\hrulefill
\vspace{-3mm}
\end{figure*}

\par To facilitate the comparison, based on (\ref{ybsreformulate}), we reformulate the RSS in (\ref{RSS}) as (\ref{RSS2}) at the top of the next page.
Since the noise terms of both $\mathcal{P} _l$ and $\mathcal{R} _l$
follow $\mathcal{C} \mathcal{N} \left( 0,\sigma _{0}^{2} \right) $, when each SC allocated with the power of $\bar{p}_{\mathrm{T}}$, the SNRs of the DSP and RSS test statistics can be respectively described by
\begin{align}
&\mathrm{snr}_{\mathrm{DSP},l}=\frac{N\bar{p}_{\mathrm{T}}\alpha _{\mathrm{CG},l}^{2}}{\sigma _{0}^{2}},
\\&\mathrm{snr}_{\mathrm{RSS},l}=\frac{\bar{p}_{\mathrm{T}}\alpha _{\mathrm{CG},l}^{2}}{\sigma _{0}^{2}},
\end{align}
which verifies that using the proposed DSP test statistic rather than the RSS test statistic obtains $N$ times SNR gain.
\end{proof}

\subsubsection{Target detection}
We design a DSP detector to determine the target presence/absence in the beam scanning area. Specifically, we first devise a binary hypothesis testing based on the DSP, which is formulated as
\begin{align}
\begin{array}{c}
	H_1:\mathcal{P} _l>\delta  ,\qquad H_0:\mathcal{P} _l\leqslant \delta  ,\\
\end{array}
\end{align}
where $H_1$/$H_0$ is the hypothesis of target presence/absence in the beam coverage, $\delta $ denotes the false alarm threshold.
\par Then, we design $\delta $ for a given false alarm rate (FAR). Under the hypothesis $H_0$, the DFBS received signals is given by
\begin{align}\label{ynotg}
y_{\mathrm{B},n,l}=\mathbf{w}_{\mathrm{B}}^{\mathrm{opt},\mathrm{T}}\mathbf{n}_{n,l}.
\end{align}
Combining (\ref{DSP_ori}) and (\ref{ynotg}), the delay spectrum $\boldsymbol{\varGamma }_l$ can be expressed as
\begin{align}
\left[ \boldsymbol{\varGamma }_l \right] _{n_{\mathrm{Q}}}=\left| \frac{1}{\sqrt{N}}\sum_{n=0}^{N-1}{\frac{\mathbf{w}_{\mathrm{B}}^{\mathrm{opt},\mathrm{T}}\mathbf{n}_{n,l}}{s_{n,l}}e^{j2\pi n\varDelta f\left( \tau _{n_{\mathrm{Q}}}+\tau _0 \right)}} \right|^2.
\end{align}
Since the inner term of $\left| \cdot \right|^2$ follows $\mathcal{C} \mathcal{N} \left( 0,\sigma _{0}^{2} \right) $, each $\left[ \boldsymbol{\varGamma }_l \right] _{n_{\mathrm{Q}}}$ follows an exponential distribution
\begin{align}
f\left( z|H_0 \right) =\frac{1}{\sigma _{0}^{2}}e^{-\frac{z}{\sigma _{0}^{2}}},z>0.
\end{align}
Hence, each $\left[ \boldsymbol{\varGamma }_l \right] _{n_{\mathrm{Q}}}$ has the FAR of
\begin{align}\label{PFAE}
p_\mathrm{FAR}&=\mathbb{P} \left( \left[ \boldsymbol{\varGamma }_l \right] _{n_{\mathrm{Q}}}>\delta  |H_0 \right) =\int_{\delta }^{\infty}{f\left( z|H_0 \right) dz}\notag
\\
&=\int_{\delta }^{\infty}{\frac{1}{\sigma _{0}^{2}}e^{-\frac{z}{\sigma _{0}^{2}}}dz}=e^{-\frac{\delta }{\sigma _{0}^{2}}}.
\end{align}
Note that, when $N_{\mathrm{Q}}=N$, the FAR of the detector is \cite{skolnik2008radar}
\begin{align}\label{PFA}
\bar{p}_\mathrm{FAR}=1-\left( 1-p_\mathrm{FAR} \right) ^N.
\end{align}
At last, combining (\ref{PFA}) and (\ref{PFAE}), the false alarm threshold of the DSP detector with  $\bar{p}_\mathrm{FAR}$ FAR is designed as
\begin{align}\label{Kappa}
\delta  =-\sigma _{0}^{2}\ln \left( 1-\sqrt[N]{1-\bar{p}_\mathrm{FAR}} \right) .
\end{align}
\subsection{BR for super-resolution direction estimation}
For the 3D HBT, its beam training result is considered to be the IRS beam with the maximum DSP in the layer $K$, i.e.,
\begin{align}
\boldsymbol{\xi }\left( \tilde{i}_K,\tilde{j}_K \right) =\mathbf{b}_{\mathrm{I}}\left( \frac{2\tilde{i}_K-1}{M}-1,\frac{2\tilde{j}_K-1}{M}-1 \right) .
\end{align}
and we determine that the target present in the coverage of $\boldsymbol{\xi }\left( \tilde{i}_K,\tilde{j}_K \right) $.
\black{However, the pre-defined codebook for the HBT has a finite resolution of $2/M$, which limits the accuracy of the target direction estimation. To cope with this issue, we utilize the 
linear interpolation algorithm \cite{hildebrand1987introduction,press2007numerical} to refine the direction of the IRS beam and achieve super-resolution direction estimation, which is called the BR approach.}

\begin{figure}[htbp]
  \centering
   \includegraphics[width=2.7in]{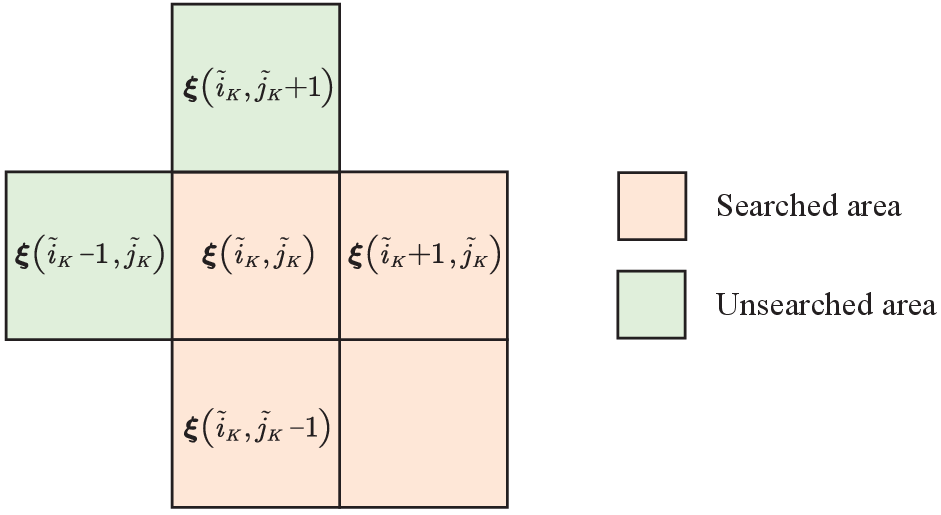}
  \caption{An illustration of BR.}
  \label{Beamrefinement}
\end{figure}
\par As portrayed in Fig.~\ref{Beamrefinement}, the IRS beam $\boldsymbol{\xi }\left( \tilde{i}_K,\tilde{j}_K \right) $ has $4$ adjacent areas which belong to the IRS beam set $\varPsi \triangleq \left\{ \boldsymbol{\xi }(\tilde{i}_K+\varDelta _i,\tilde{j}_K+\varDelta _j)|\varDelta _i,\varDelta _j\in \left\{ -1,1 \right\} \right\} $, where $2$ areas have already been searched in the last stage of the HBT while the other $2$ areas have not. In the BR approach, we first scan the $2$ unsearched areas in sequence and calculate their corresponding DSPs. Then, by using the linear interpolation algorithm, the azimuth/elevation index can be updated as
\begin{align}
\small{i_{\mathrm{BR}}=\frac{\sum\nolimits_{\varDelta _i=-1}^1{\left( \tilde{i}_K+\varDelta _i \right) \mathcal{P} \left( \tilde{i}_K+\varDelta _i,\tilde{j}_K \right)}}{\sum\nolimits_{\varDelta _i=-1}^1{\mathcal{P} \left( \tilde{i}_K+\varDelta _i,\tilde{j}_K \right)}},}
\\
\small{j_{\mathrm{BR}}=\frac{\sum\nolimits_{\varDelta _j=-1}^1{\left( \tilde{j}_K+\varDelta _j \right) \mathcal{P} \left( \tilde{i}_K,\tilde{j}_K+\varDelta _j \right)}}{\sum\nolimits_{\varDelta _j=-1}^1{\mathcal{P} \left( \tilde{i}_K,\tilde{j}_K+\varDelta _j \right)}}.}
\end{align}
\par As such, the IRS beam is refined to be
\begin{align}\label{IRSfinal}
\boldsymbol{\xi }_{\mathrm{BR}}=\mathbf{b}_{\mathrm{I}}\left( \frac{2i_{\mathrm{BR}}-1}{M}-1,\frac{2j_{\mathrm{BR}}-1}{M}-1 \right) .
\end{align}
Recall from (\ref{anglerelation}) that $\boldsymbol{\xi }_{\mathrm{BR}}$ maximizes the beam gain along the direction 
\begin{align}
\small{\left( \frac{2i_{\mathrm{BR}}-1}{M}-1+u_{\mathrm{B}2\mathrm{R}}^{\mathrm{A}}\,\,,\frac{2j_{\mathrm{BR}}-1}{M}-1+v_{\mathrm{B}2\mathrm{R}}^{\mathrm{A}} \right) ,}
\end{align}
which is estimated as the target direction $\left( \hat{u}_{\mathrm{R}2\mathrm{G}}^{\mathrm{D}},\hat{v}_{\mathrm{R}2\mathrm{G}}^{\mathrm{D}} \right) $.
\vspace{3mm}
\par \black{\textbf{Complexity analysis: }In the proposed IRS 3D beam training scheme, the IRS first conducts 3D HBT including $K=\log _2M$ training stages, during each of which $4$ IRS beams are trained. Then, the BR is conducted, where the IRS performs beamforming to scan the $2$ unsearched areas around the beam training result of 3D HBT. Hence, the proposed beam training scheme requires the implementation overhead of $\left( 2+4\log _2M \right) $ OFDM symbols, which is much lower than the overhead of the EBS scheme (i.e., $M^2$). Moreover, during the training duration of each IRS beam, the DFBS calculates the DSP of the received echo signals, where the corresponding delay spectrum has $N_{\mathrm{Q}}$ elements, each of which involves $\mathcal{O} \left( N \right) $ computational complexity. Hence, the computational complexity of the proposed beam training scheme is $\mathcal{O} \left( \left( 2+4\log _2M \right) N_{\mathrm{Q}}N \right) $.}
\color{black}
\vspace{-2mm}
\subsection{Performance Bound for direction estimation}
Based on (\ref{hI2T}) and (\ref{yBS}), the DFBS received echo signals during the CGS period can be expressed as
\begin{align}
y_{\mathrm{B},n,l}=\sqrt{p_{\mathrm{T},n}}\kappa _lB_{n,l}\left( u_{\mathrm{R}2\mathrm{G}}^{\mathrm{D}},v_{\mathrm{R}2\mathrm{G}}^{\mathrm{D}} \right) s_{n,l}+\bar{n}_{n,l},
\end{align}
where
\begin{align}
&\kappa_l= \alpha _{\mathrm{G}}\alpha _{\mathrm{R}2\mathrm{G},n}^{2}e^{j2\pi lT_{\mathrm{O}}f_{\mathrm{D}}},
\\
&B_{n,l}\left( u_{\mathrm{R}2\mathrm{G}}^{\mathrm{D}},v_{\mathrm{R}2\mathrm{G}}^{\mathrm{D}} \right) =\mathbf{C}_{n,l}^{\mathrm{T}}\mathbf{C}_{n,l},
\\
&\mathbf{C}_{n,l}=\mathbf{b}_{\mathrm{I}}^{\mathrm{H}}\left( u_{\mathrm{R}2\mathrm{G}}^{\mathrm{D}},v_{\mathrm{R}2\mathrm{G}}^{\mathrm{D}} \right) \mathrm{diag}\left( \boldsymbol{\xi }\left( l \right) \right) \mathbf{G}_{\mathrm{B}2\mathrm{R},n}\mathbf{w}_{\mathrm{B}}^{\mathrm{opt}}.
\end{align}
\par We denote $\boldsymbol{\varPsi }=\left[ \boldsymbol{\eta }^{\mathrm{T}},\boldsymbol{\kappa }_{l}^{\mathrm{T}} \right] ^{\mathrm{T}}\in \mathbb{R} ^{4\times 1}$ as the vector of unknown parameters to be estimated, where $\boldsymbol{\eta }=\left[ u_{\mathrm{R}2\mathrm{G}}^{\mathrm{D}},v_{\mathrm{R}2\mathrm{G}}^{\mathrm{D}} \right] ^{\mathrm{T}}$ and $\boldsymbol{\kappa }_l=\left[ \mathrm{Re}\left\{ \kappa _l \right\} ,\mathrm{Im}\left\{ \kappa _l \right\} \right] ^{\mathrm{T}}$.
\par Then, we stack the DFBS received echo signals on all SCs during the training duration of the $l$-th training beam into an $N\times1$ vector given by
\begin{align}
\mathbf{y}_{\mathrm{CG},l}&= \left[ \begin{array}{c}
	\sqrt{p_{\mathrm{T},1}}\kappa _lB_l\left( u_{\mathrm{R}2\mathrm{G}}^{\mathrm{D}},v_{\mathrm{R}2\mathrm{G}}^{\mathrm{D}} \right) s_{1,l}\\
	\vdots\\
	\sqrt{p_{\mathrm{T},N}}\kappa _lB_l\left( u_{\mathrm{R}2\mathrm{G}}^{\mathrm{D}},v_{\mathrm{R}2\mathrm{G}}^{\mathrm{D}} \right) s_{N,l}\\
\end{array} \right] +\left[ \begin{array}{c}
	\bar{n}_{1,l}\\
	\vdots\\
	\bar{n}_{N,l}\\
\end{array} \right] \notag
\\
&\triangleq \bar{\mathbf{r}}_l+\bar{\mathbf{n}}_l.
\end{align}
where $\bar{\mathbf{n}}_l\sim \mathcal{C} \mathcal{N} \left( \mathbf{0},\mathbf{\Sigma } \right) $ with $\mathbf{\Sigma }=\sigma _{0}^{2}\mathbf{I}_N$.
\par Next, we define $\mathbf{J}\in \mathbb{R} ^{4\times 4}$ as the Fisher information matrix (FIM) for estimating $\boldsymbol{\varPsi }$, whose $(i,j)$-th element is \cite{kay1993fundamentals}
\begin{align}\label{FIM}
\left[ \mathbf{J} \right] _{i,j}&=\mathrm{tr}\left( \mathbf{\Sigma }^{-1}\frac{\partial \mathbf{\Sigma }}{\partial \boldsymbol{\varPsi }_i}\mathbf{\Sigma }^{-1}\frac{\partial \mathbf{\Sigma }}{\partial \boldsymbol{\varPsi }_j} \right) +2\mathrm{Re}\left\{ \frac{\partial \bar{\mathbf{r}}_{l}^{\mathrm{H}}}{\partial \boldsymbol{\varPsi }_i}\mathbf{\Sigma }^{-1}\frac{\partial \bar{\mathbf{r}}_l}{\partial \boldsymbol{\varPsi }_j} \right\} \notag
\\
&=\frac{2}{\sigma _{0}^{2}}\mathrm{Re}\left\{ \frac{\partial \bar{\mathbf{r}}_{l}^{\mathrm{H}}}{\partial \boldsymbol{\varPsi }_i}\frac{\partial \bar{\mathbf{r}}_l}{\partial \boldsymbol{\varPsi }_j} \right\} ,i,j\in \left\{ 1,2,3,4 \right\} .
\end{align}
\par Finally, following the derivation in \cite{1703855}, we write $\mathbf{J}$  as
\begin{align}
\mathbf{J}=\left[ \begin{matrix}
	\mathbf{J}_{\boldsymbol{\eta \eta }}&		\mathbf{J}_{\boldsymbol{\eta \kappa }_l}\\
	\mathbf{J}_{\boldsymbol{\kappa }_l\boldsymbol{\eta }}&		\mathbf{J}_{\boldsymbol{\kappa }_l\boldsymbol{\kappa }_l}\\
\end{matrix} \right] ,
\end{align}
and the CRLB for estimating $\boldsymbol{\eta }$ can be calculated as 
\begin{align}
\mathrm{CRLB}\left( \boldsymbol{\eta } \right) =\left[ \mathbf{J}_{\boldsymbol{\eta \eta }}-\mathbf{J}_{\boldsymbol{\eta \kappa }_l}\mathbf{J}_{\boldsymbol{\kappa }_l\boldsymbol{\kappa }_l}^{-1}\mathbf{J}_{\boldsymbol{\kappa }_l\boldsymbol{\eta }}^{\mathrm{T}} \right] ^{-1}.
\end{align}
\par Since the estimation of horizontal and vertical directions (i.e., $u_{\mathrm{R}2\mathrm{G}}^{\mathrm{D}}$ and $v_{\mathrm{R}2\mathrm{G}}^{\mathrm{D}}$) are decoupled \cite{ShaoXiaodanTarget}, their estimation performance can be separately analyzed, following the similar method in \cite{SongXianxinIntelligent}.
\color{black}
\section{Target R\&V Estimation via DD Estimation}\label{section4}
\black{During the FGS period, we estimate the R\&V of the NLOS target, by extracting DD information from the IRS reflected target echo signals, where the IRS beamforming is configured as the beam training result $\boldsymbol{\xi }_{\mathrm{BR}}$ obtained via the CGS, which maximizes the SNR of the DFBS received echo signals.}
\par According to (\ref{HB2I}), (\ref{hI2T}), and (\ref{yBS}), we pack the DFBS received echo signals during the FGS period into an $N\times T_{\mathrm{FG}}$-matrix, whose $(n,l)$-th element is described by (\ref{YFS}) at the top of this page,
\begin{figure*}[t]
\begin{equation}\label{YFS}
 \begin{gathered}
  \begin{aligned}
\left[ \mathbf{Y}_{\mathrm{FG}} \right] _{n,l}&=\sqrt{p_{\mathrm{T},n}}\alpha _{\mathrm{G}}\mathbf{w}_{\mathrm{B}}^{\mathrm{opt},\mathrm{T}}\bar{\mathbf{G}}_n\left( \mathrm{diag}\left( \boldsymbol{\xi }_{\mathrm{BR}} \right) \right) \mathbf{w}_{\mathrm{B}}^{\mathrm{opt}}s_{n,\left( l+T_{\mathrm{CG}} \right)}e^{j2\pi \left( l+T_{\mathrm{CG}} \right) T_{\mathrm{O}}f_{\mathrm{D}}}+\bar{n}_{n,l}
\\
&=\sqrt{p_{\mathrm{T},n}}\alpha _{\mathrm{FG}}e^{j\beta _{\mathrm{FG}}}e^{-j2\pi n\varDelta f\left( \tau _0+\tau  \right)}s_{n,\left( l+T_{\mathrm{CG}} \right)}e^{j2\pi lT_{\mathrm{O}}f_{\mathrm{D}}}+\bar{n}_{n,l},n\in \mathcal{N} ,l\in \mathcal{L} _{\mathrm{FG}}\triangleq \left\{ 0,\cdots ,T_{\mathrm{FG}}-1 \right\} ,
 \end{aligned}
 \end{gathered}
\end{equation}
\hrulefill
\vspace{-3mm}
\end{figure*}
where 
\begin{align}
&\alpha _{\mathrm{FG}}\triangleq N_{\mathrm{B}}\alpha _{\mathrm{G}}a_{\mathrm{R}2\mathrm{G}}^{2}a_{\mathrm{B}2\mathrm{R}}^{2}\left| h_{\boldsymbol{\xi }_{\mathrm{BR}}}^{2} \right|,
\\
&\beta _{\mathrm{FG}}\triangleq \mathrm{angle}\left( h_{\boldsymbol{\xi }_{\mathrm{BR}}}^{2}\! \right) +2\pi \left( T_{\mathrm{CG}}T_{\mathrm{O}}f_{\mathrm{D}}-f_c\left( \tau _0+\tau  \right) \right) ,
\end{align}
where $h_{\boldsymbol{\xi }_{\mathrm{BR}}}\!=\!\mathbf{b}_{\mathrm{I}}^{\mathrm{H}}\left( u_{\mathrm{R}2\mathrm{G}}^{\mathrm{D}},v_{\mathrm{R}2\mathrm{G}}^{\mathrm{D}} \right) \mathrm{diag}\left( \boldsymbol{\xi }_{\mathrm{BR}} \right) \mathbf{b}_{\mathrm{I}}\left( u_{\mathrm{B}2\mathrm{R}}^{\mathrm{A}},v_{\mathrm{B}2\mathrm{R}}^{\mathrm{A}} \right) $. 
\par Then, we define an auxiliary matrix $\mathbf{F}\in \mathbb{C} ^{N\times T_{\mathrm{FG}}}$, whose $(n,l)$-th element is described by
\begin{align}
&\left[ \mathbf{F} \right] _{n,l}=e^{j2\pi n\varDelta f\tau _0}\frac{\left[ \mathbf{Y}_{\mathrm{FG}} \right] _{n,l}}{s_{n,\left( l+T_{\mathrm{CG}} \right)}}
\\
&=\sqrt{p_{\mathrm{T},n}}\alpha _{\mathrm{FG}}e^{j\beta _{\mathrm{FG}}}e^{-j2\pi n\varDelta f\tau }e^{j2\pi lT_{\mathrm{O}}f_{\mathrm{D}}}\!+\!\frac{e^{j2\pi n\varDelta f\tau _0}\bar{n}_{n,l}}{s_{n,\left( l+T_{\mathrm{CG}} \right)}}.\notag
\end{align}
Below we leverage $\mathbf{F}$ for estimating the R\&V of the target.
\subsection{Maximum Likelihood Estimation}
First, stacking the unknown parameters into a parameter vector $\boldsymbol{\theta }$ yields
\begin{align}
\boldsymbol{\theta }=\left[ \alpha _{\mathrm{FG}},\beta _{\mathrm{FG}},\tau  ,f_D \right] ^{\mathrm{T}}.
\end{align}
\par Since the noise term in $\mathbf{F}$ is independent, the log-likelihood function is described by
\begin{align}
\ell \left( \mathbf{F}|\boldsymbol{\theta } \right) =\log \prod_{n=0}^{N-1}{\prod_{l=0}^{T_{\mathrm{FG}}-1}{\left( \frac{1}{\pi \sigma _{0}^{2}}e^{-\frac{D_{n,l}}{\sigma _{0}^{2}}} \right)}}\notag
\\
=\sum_{n=0}^{N-1}{\sum_{l=0}^{T_{\mathrm{FG}}-1}{\left( -\log \pi \sigma _{0}^{2}-\frac{D_{n,l}}{\sigma _{0}^{2}} \right)}},
\end{align}
where $D_{n,l}\!=\!\left| \left[ \mathbf{F} \right] _{n,l}\!-\!\sqrt{p_{\mathrm{T},n}}\alpha _{\mathrm{FG}}e^{j\left( 2\pi \left( lT_{\mathrm{O}}f_D\!-n\varDelta f\tau  \right) +\beta _{\mathrm{FG}} \right)} \right|^2$.
As such, the maximum likelihood estimation (MLE) of $\boldsymbol{\theta }$ is
\begin{align}
\hat{\boldsymbol{\theta}}_{\mathrm{ML}}=\,\,\underset{\boldsymbol{\theta }}{\mathrm{arg}\max}\,\,\ell \left( \mathbf{F}|\boldsymbol{\theta } \right) .
\end{align}
For the log-likelihood function, its first term $-\log \pi \sigma _{0}^{2}$ and the factor of the second term $1/\sigma _{0}^{2}$ are independent of ${\bm \theta}$ and have not affect on the function maximization. Hence, by ignoring $-\log \pi \sigma _{0}^{2}$ and $1/\sigma _{0}^{2}$, the objective function is given by
\begin{align}
&\sum_{n=0}^{N-1}\!{\sum_{l=0}^{T_{\mathrm{FG}}-1}{\!\!-\!D_{n,l}}}=\!\sum_{n=0}^{N-1}{\sum_{l=0}^{T_{\mathrm{FG}}-1}{\!\left( -p_{\mathrm{T},n}\alpha _{\mathrm{FG}}^{2}\!-\left| \left[ \mathbf{F} \right] _{n,l} \right|^2 \right.}}
\\
&\left.+2\alpha _{\mathrm{FG}}\sqrt{p_{\mathrm{T},n}}\mathrm{Re}\left\{ \left[ \mathbf{F} \right] _{n,l}e^{-j\left( 2\pi \left( lT_{\mathrm{O}}f_D-n\varDelta f\tau  \right) +\beta _{\mathrm{FG}} \right)} \right\} \right) .\notag
\end{align}
Similarly, we further ignore those terms that are independent of ${\bm \theta}$, and simplify the objective function to 
\begin{align}\label{likelihood}
\sum_{n=0}^{N-1}{\sum_{l=0}^{T_{\mathrm{FG}}-1}{\mathrm{Re}\left\{ \left[ \mathbf{F} \right] _{n,l}e^{-j\left( 2\pi \left( lT_{\mathrm{O}}f_D-n\varDelta f\tau  \right) +\beta_{\mathrm{FG}} \right)} \right\}}},
\end{align}
which is denoted as $\hat{\ell}\left( \mathbf{F}|\boldsymbol{\theta } \right)$, and $\hat{\boldsymbol{\theta}}_{\mathrm{ML}}$ can be obtained from
\begin{align}
\boldsymbol{\hat{\theta}}_{\mathrm{ML}}=\,\,\underset{\boldsymbol{\theta }}{\mathrm{arg}\max}\,\,\hat{\ell}\left( \mathbf{F}|\boldsymbol{\theta } \right) .
\end{align}
\subsection{DD Estimation}
\par For parameter estimation, we first reformulate (\ref{likelihood}) as
\begin{align}\label{relikelihood}
\hat{\ell}\left( \mathbf{F}|\boldsymbol{\theta } \right) =&\mathrm{Re}\left\{ e^{j\beta _{\mathrm{FG}}} \right. 
\\
&\times \left. \sum_{n=0}^{N-1}{\left( \sum_{l=0}^{T_{\mathrm{FG}}-1}{\!\left[ \mathbf{F} \right] _{n,l}e^{-j2\pi lT_{\mathrm{O}}f_D}} \right) e^{j2\pi n\varDelta f\tau }} \right\} .\notag
\end{align}
Observing that  (\ref{relikelihood}) exhibits a high degree of resemblance with the Discrete Fourier transform (DFT). Hence, by defining the delay quantization vector $\boldsymbol{\tau }$ and Doppler shift quantization vector $\boldsymbol{f}_{\mathrm{D}}$ as
\begin{align}
&\boldsymbol{\tau }=\frac{\tau _{\max}}{N_{\mathrm{R}}}\left[ 1,\cdots ,n_{\mathrm{D}},\cdots ,\left( N_{\mathrm{R}}-1 \right) \right] ,\label{Qtau}
\\
&\boldsymbol{f}_{\mathrm{D}}=\frac{f_{\mathrm{D},\max}}{T_{\mathrm{V}}}\left[ -T_{\mathrm{V}},\cdots ,l_{\mathrm{D}},\cdots ,\left( T_{\mathrm{V}}-1 \right) \right] ,
\label{Qfd}
\end{align}
where $\tau _{\max}=\frac{1}{2\varDelta f}$ and $f_{\mathrm{D},\max}=\frac{1}{2T_{\mathrm{O}}}$ respectively denote the maximum unambiguous time delay and the maximum unambiguous Doppler shift\cite{richards2014fundamentals}. Then, the MLEs of $\tau $/$f_{\mathrm{D}}$ is that which maximizes the DD spectrum $\mathbf{\Lambda }\in \mathbb{R} ^{N_{\mathrm{R}}\times 2T_{\mathrm{V}}}$ \cite{RifeSingle}, whose $(n_{\mathrm{D}},l_{\mathrm{D}})$-th element is described by (\ref{DDSpectrum}) at the top of this page.
\begin{figure*}[htbp]
\begin{equation}\label{DDSpectrum}
 \begin{gathered}
  \begin{aligned}
\left[ \mathbf{\Lambda } \right] _{n_{\mathrm{D}},l_{\mathrm{D}}}=\frac{1}{NT_{\mathrm{FG}}}\left| \sum_{n=0}^{N-1}{\left( \sum_{l=0}^{T_{\mathrm{FG}}-1}{\left[ \mathbf{F} \right] _{n,l}e^{-j2\pi lT_{\mathrm{O}}\boldsymbol{f}_{\mathrm{D}}\left( l_{\mathrm{D}} \right)}} \right) e^{j2\pi n\varDelta f\boldsymbol{\tau }\left( n_{\mathrm{D}} \right)}} \right|^2.
 \end{aligned}
 \end{gathered}
\end{equation}
\hrulefill
\end{figure*}
\par However, for the above DD estimation process, there are $2N_{\mathrm{R}}T_{\mathrm{V}}$ DD spectrum elements to be calculated in order to achieve $\frac{\tau _{\max}}{N_{\mathrm{R}}}$ delay quantization resolution and $\frac{f_{\mathrm{D},\max}}{T_{\mathrm{V}}}$ Doppler shift quantization resolution, leading to a high computational complexity when requiring high-accuracy DD estimation. To reduce the complexity, we design a hierarchical R\&V estimation (H-R\&VE) algorithm as shown in Algorithm 1, which achieves $\frac{\tau _{\max}}{N_{\mathrm{R}}^{I}}$ delay quantization resolution and $\frac{f_{\mathrm{D},\max}}{T_{\mathrm{V}}^{I}}
$ Doppler shift quantization resolution while only requiring $\left( 4I+2 \right) N_{\mathrm{R}}T_{\mathrm{V}}$ calculations of DD spectrum elements, where $I$ is the iteration number.
\vspace{-3mm}
\begin{algorithm}[htbp]
    \KwIn{$\mathbf{S}$}
    \SetAlgoLined 
    \DontPrintSemicolon
	\caption{H-R\&VE Algorithm}
    Initialize $\boldsymbol{\tau }_0$ and $\boldsymbol{f}_{\mathrm{D},0}$ based on (\ref{Qtau}) and (\ref{Qfd}).\;
    
    \For{$i=1,\cdots,I$}{

    Calculate the DD spectrum $\mathbf{\Lambda }$ corresponding to $\boldsymbol{\tau }_{i-1}$ and $\boldsymbol{f}_{\mathrm{D},\left( i-1 \right)}$ based on (\ref{DDSpectrum}).\;

    Calculate $( \hat{n}_{\mathrm{D}},\hat{l}_{\mathrm{D}} ) =\underset{\left( n_{\mathrm{D}},l_{\mathrm{D}} \right)}{\mathrm{arg}\max}\,\,\left[ \mathbf{\Lambda } \right] _{n_{\mathrm{D}},l_{\mathrm{D}}}$.\;

    Update $\boldsymbol{\tau }_i=\boldsymbol{\tau }_{i-1}\left( \hat{n}_{\mathrm{D}} \right) +\frac{\tau _{\max}}{N_{\mathrm{R}}^{i+1}}\left[ -N_{\mathrm{R}},\cdots ,n_{\mathrm{D}},\cdots ,N_{\mathrm{R}} \right] $.\;

    Update $\boldsymbol{f}_{\mathrm{D},i}=\boldsymbol{f}_{\mathrm{D},\left( i-1 \right)}( \hat{l}_{\mathrm{D}} ) +\frac{f_{\mathrm{D},\max}}{T_{\mathrm{V}}^{i+1}}\left[ -T_{\mathrm{V}},\cdots ,l_{\mathrm{D}},\cdots ,T_{\mathrm{V}} \right] $.\;
    }
    Calculate the DD spectrum $\mathbf{\Lambda }$ corresponding to $\boldsymbol{\tau }_{I}$ and $\boldsymbol{f}_{\mathrm{D},I}$ based on (\ref{DDSpectrum}).\;

    Calculate $( \hat{n}_{\mathrm{D}},\hat{l}_{\mathrm{D}} ) =\underset{\left( n_{\mathrm{D}},l_{\mathrm{D}} \right)}{\mathrm{arg}\max}\,\,\left[ \mathbf{\Lambda } \right] _{n_{\mathrm{D}},l_{\mathrm{D}}}$.\;
    
    Calculate $\hat{\tau}=\boldsymbol{\tau }_I\left( \hat{n}_{\mathrm{D}} \right) $ and $\hat{f}_{\mathrm{D}}=\boldsymbol{f}_{\mathrm{D},I}\left( \hat{l}_{\mathrm{D}} \right) $.\;
\KwOut{$\hat{\tau}$, $\hat{f}_{\mathrm{D}}$}
\end{algorithm}

\par After obtaining $\hat{\tau}$ and $\hat{f}_{\mathrm{D}}$, the target R\&V is estimated as
\begin{align}
&\hat{d}_{\mathrm{R}2\mathrm{G}}=\frac{c\hat{\tau}}{2},\label{estimater}\\
&\hat{\mathrm{v}}_{\mathrm{G}}=\frac{\hat{f}_{\mathrm{D}}\lambda}{2}.\label{estimatev}
\end{align}
Moreover, combining the obtained $\hat{d}_{\mathrm{R}2\mathrm{G}}$ in the FGS period and $\left( \hat{u}_{\mathrm{R}2\mathrm{G}}^{\mathrm{D}},\hat{v}_{\mathrm{R}2\mathrm{G}}^{\mathrm{D}} \right) $ in the CGS period, we can estimate the target location as $\hat{\mathbf{q}}_{\mathrm{G}}=\left[ \hat{x}_{\mathrm{G}},\hat{y}_{\mathrm{G}},\hat{z}_{\mathrm{G}} \right] ^{\mathrm{T}}$, where
\begin{align}
&\hat{y}_{\mathrm{G}}=y_{\mathrm{IRS}}-\hat{u}_{\mathrm{R}2\mathrm{G}}^{\mathrm{D}}\hat{d}_{\mathrm{R}2\mathrm{G}},\label{estimatestart}
\\
&\hat{z}_{\mathrm{G}}=z_{\mathrm{IRS}}-\hat{v}_{\mathrm{R}2\mathrm{G}}^{\mathrm{D}}\hat{d}_{\mathrm{R}2\mathrm{G}},
\\
&\hat{x}_{\mathrm{G}}\!=\!x_{\mathrm{IRS}}\!+\!\sqrt{\hat{d}_{\mathrm{R}2\mathrm{G}}^{2}\!-\!\left( y_{\mathrm{IRS}}\!-\!y_{\mathrm{G}} \right) ^2\!-\!\left( z_{\mathrm{IRS}}\!-\!z_{\mathrm{G}} \right) ^2},\label{estimateend}
\end{align}
where $\left[ x_{\mathrm{IRS}},y_{\mathrm{IRS}},z_{\mathrm{IRS}} \right] ^{\mathrm{T}}$ denotes the IRS location.
\par Finally, the main procedures of the proposed IRS-assisted NLOS target sensing strategy are summarized in Algorithm 2.
\vspace{-3mm}
\begin{algorithm}[htbp]
\color{black}{
\SetKw{CGS}{CGS period}
\SetKw{FGS}{FGS period}
    \SetAlgoLined 
    \DontPrintSemicolon
	\caption{IRS-assisted NLOS Target Sensing Strategy}
	\KwIn{$\mathbf{S}$}
    Initialize the hierarchical codebook $\bm{\mathcal{W}}$ according to Section III. B, and initialize the false alarm threshold $\delta $ based on (\ref{Kappa}).
    
    \CGS:{}\;
    Conduct DSP-based 3D HBT for target detection and direction estimation according to Section III. C, which yields $\boldsymbol{\xi }\left( \tilde{i}_K,\tilde{j}_K \right)$.\;
    
    Conduct BR of $\boldsymbol{\xi }\left( \tilde{i}_K,\tilde{j}_K \right)$ according to Section III. D, which yields the refined beam $\boldsymbol{\xi }_{\mathrm{BR}}$ and the estimated target direction $\left( \hat{u}_{\mathrm{R}2\mathrm{G}}^{\mathrm{D}},\hat{v}_{\mathrm{R}2\mathrm{G}}^{\mathrm{D}} \right) $.\;
    
    \FGS:{}\;
    
    Design the IRS beamforming as $\boldsymbol{\xi }_{\mathrm{BR}}$.
      
    Conduct DD estimation on the received echo signals according to Algorithm 1, which yields $\hat{\tau}$ and $\hat{f}_{\mathrm{D}}$.
    
    Calculate the target R\&V based on (\ref{estimater}) and (\ref{estimatev}), which yields $\hat{d}_{\mathrm{R}2\mathrm{G}}$ and $\hat{\mathrm{v}}_{\mathrm{G}}$.
    
    Localize the target based on (\ref{estimatestart}) to (\ref{estimateend}), which yields $\hat{\mathbf{q}}_{\mathrm{G}}$.
    
    \KwOut{$\left( \hat{u}_{\mathrm{R}2\mathrm{G}}^{\mathrm{D}},\hat{v}_{\mathrm{R}2\mathrm{G}}^{\mathrm{D}} \right) $, $\hat{d}_{\mathrm{R}2\mathrm{G}}$, $\hat{\mathrm{v}}_{\mathrm{G}}$, $\hat{\mathbf{q}}_{\mathrm{G}}$.}
}\end{algorithm}
\color{black}
\vspace{-2mm}
\subsection{Performance Bounds for R\&V Estimation}
To obtain more intuitive insights, we consider the special case where each SC is allocated with the same power of $\bar{p}_{\mathrm{T}}$. According to (\ref{YFS}), the received SNR during the FGS period (i.e., $\mathbf{Y}_{\mathrm{FG}}$) is 
\begin{align}
\mathrm{snr}_{\mathrm{FG}}=\frac{\bar{p}_{\mathrm{T}}N_{\mathrm{B}}^{2}\zeta _{\mathrm{G}}^{2}a_{\mathrm{R}2\mathrm{G}}^{4}a_{\mathrm{B}2\mathrm{R}}^{4}\left| h_{\boldsymbol{\xi }_{\mathrm{BR}}}^{2} \right|^2}{\sigma _{0}^{2}}.
\end{align}
Following the derivation in \cite{braun2011single}, the CRLBs for R\&V estimation are given by
\begin{align}
&\mathrm{CRLB}\left( \hat{d}_{\mathrm{R}2\mathrm{G}} \right) \!\geqslant \!\frac{6}{\mathrm{snr}_{\mathrm{FG}}\left( N^2-1 \right) NT_{\mathrm{FG}}}\left( \frac{c}{4\pi \varDelta f} \right) ^2,
\\
&\mathrm{CRLB}\left( \hat{\mathrm{v}}_{\mathrm{G}} \right) \geqslant \frac{6}{\mathrm{snr}_{\mathrm{FG}}\left( T_{\mathrm{FG}}^{2}-1 \right) T_{\mathrm{FG}}N}\left( \frac{c}{4\pi T_{\mathrm{O}}f_c} \right) ^2,
\end{align}
which shows that the accuracy of R\&V estimation during the FGS period can be improved by optimizing the IRS phase shift vector $\boldsymbol{\xi }_{\mathrm{BR}}$, adding more time-frequency resources $N T_{\mathrm{FG}}$, as well as allocating more transmit power $\bar{p}_{\mathrm{T}}$.

\section{Extension To the General Multi-Target Case}\label{section5}
The proposed IRS-assisted NLOS target sensing system can be extended to the more general case with $A$ targets. 
\par The multi-target detection and direction estimation during the CGS period mainly includes the following four steps:
\begin{enumerate}
    \item Design the hierarchical codebook of IRS beamforming (See Section III. B).
    \item Conduct a DSP-based multi-target 3D HBT which has $K$ stages. Specifically, in stage $1$, the IRS generates $4$ beams in sequence, by applying the codewords of the first layer. Meanwhile, according to the received echo signals, the DFBS calculates the DSP corresponding to each training beam. In stage $k$ ($k>1$), those beams with DSP higher than the false alarm threshold $\delta $ in the previous stage are selected, and the coverage area of each selected beam is equally divided into four parts to be scanned further, by applying the $k$-th layer's codewords.
    \item Determine the presences and directions of $A$ targets.
    Note that, in the multi-target case, the selected areas in the last stage can not be directly determined with target presence due to the existence of sidelobe gain. Moreover, combining (\ref{aR2G}), (\ref{alphaCG}), and (\ref{DSP_ori}), we observe that the DSP is linearly proportional to ${d_{\mathrm{R}2\mathrm{G}}}^{-2\epsilon _{\mathrm{R}2\mathrm{G}}}$, which indicates that the difference in target distance results in substantial variation in the DSP of target echo signals, rendering the challenge of determining multi-target presence/absence based on the DSP metric. To cope with these issues, we define the distance-normalized delay spectrum as
    \begin{align}
    \left[ \bar{\boldsymbol{\varGamma}}_l \right] _{n_{\mathrm{Q}}}=\left[ \boldsymbol{\varGamma }_l \right] _{n_{\mathrm{Q}}}\left( \frac{c\tau _{n_{\mathrm{Q}}}}{2} \right) ^{2\epsilon _{\mathrm{R}2\mathrm{G}}},n_{\mathrm{Q}}\in \mathcal{N} _{\mathrm{Q}}.
    \end{align}
    Then, we calculate the distance-normalized DSP corresponding to the selected areas in the last stage, and determine the $A$ areas corresponding to the $A$ largest distance-normalized DSPs with target presence.
    \item Conduct BR for the areas that are determined with target presence, and estimate the directions of $A$ targets based on the directions of the refined beams.
\end{enumerate}
\par For multi-target R\&V estimation, during the FGS period, the IRS beamforming is sequentially configured following the $A$ refined beams obtained during the CGS period, where each beam lasts for $T_{\mathrm{FG}}/A$ OFDM symbols. Then, by applying Algorithm 1, the DFBS processes the echo signals received during the training duration of each beam, which yields the R\&Vs of the NLOS targets within the corresponding beam coverage (See Section IV).
\color{black}
\section{Numerical Results}\label{section6}
This section verifies the effectiveness of the proposed IRS-assisted NLOS target sensing framework via simulations. The DFBS and the IRS are respectively deployed at $\mathbf{q}_{\mathrm{B}}=\left[ 35,-20,10 \right] ^{\mathrm{T}}$ m and $\mathbf{q}_{\mathrm{IRS}}=\left[ 0,0,10 \right] ^{\mathrm{T}}$ m. The NLOS target is randomly generated, which satisfies $u_{\mathrm{R}2\mathrm{G}}^{\mathrm{D}},v_{\mathrm{R}2\mathrm{G}}^{\mathrm{D}}\in \left[ -\frac{1}{2},\frac{1}{2} \right] $ and $d_{\mathrm{R}2\mathrm{G}}=10$ m. The total transmit power is $25$ dBm, which is uniformly allocated among all SCs. If not specified otherwise, we set the system parameters as given in TABLE~\ref{table1}.
\vspace{-5mm}
\begin{table}[htbp]
\centering
\caption{System Parameters}
\vspace{5mm}
\label{table1}
\begin{tabular}{|c|c|c|c|}
\hline
\textbf{Parameter} &\textbf{Value} &\textbf{Parameter} &\textbf{Value} \\ \hline
$f_c$ & $28.5$ & $B$ & $100$ MHz \\ \hline
$\varDelta f$ & $120$ KHz & $N$ & $833$ \\ \hline
$T_{\mathrm{O}}$ & $8.33$ $\mu$s& $T_{\mathrm{cp}}$ & $0.58$ $\mu$s\\ \hline
$M$ & $64$ & $N_{\mathrm{B}}$ & $64$ \\ \hline
$T_{\mathrm{CG}}$ & $24$ &  $T_{\mathrm{FG}}$ &  $42$ \\\hline
$\sigma _{0}^{2}$ & $-123.2$ dBm&  $\mathrm{v}_{\mathrm{G}}$ &  $20$ m/s\\\hline
$\epsilon _{\mathrm{B}2\mathrm{R}}$ & $2.1$ &  $\epsilon _{\mathrm{R}2\mathrm{G}}$ &  $2.2$ \\\hline
$\zeta _{\mathrm{G}}^{2}$ & $1$ &  $\bar{p}_\mathrm{FAR}$ &  $0.01$ \\\hline
$N_{\mathrm{Q}}$ & $833$ & $I$  & $10$  \\\hline
$N_{\mathrm{R}}$ & $100$ & $T_{\mathrm{V}}$  & $100$  \\\hline
\end{tabular}
\end{table}
\vspace{-5mm}
\subsection{Performance of Target Detection and Direction Estimation in the CGS Period}
\par First, Fig.~\ref{CS_N} contrasts the performance of the 3D HBT with the DSP detector and the RSS detector in terms of target detection success rate, which is depicted via the probability of the target being located within the coverage area of $\boldsymbol{\xi }( \hat{i}_K,\hat{j}_K ) $. As analyzed in Proposition 1, the success rate of the proposed DSP-based HBT surpasses the RSS-based HBT, which is attributed to its capability of exploiting the correlation between multiple SCs for noise suppression. Moreover, we can observe that adopting more SCs enhances the success rates of both two HBT schemes, particularly with fewer SCs. This improvement stems from the noise mitigation advantage gained by gathering observation data across more SCs. 
\begin{figure}[htbp]
  \centering
   \includegraphics[width=2.8in]{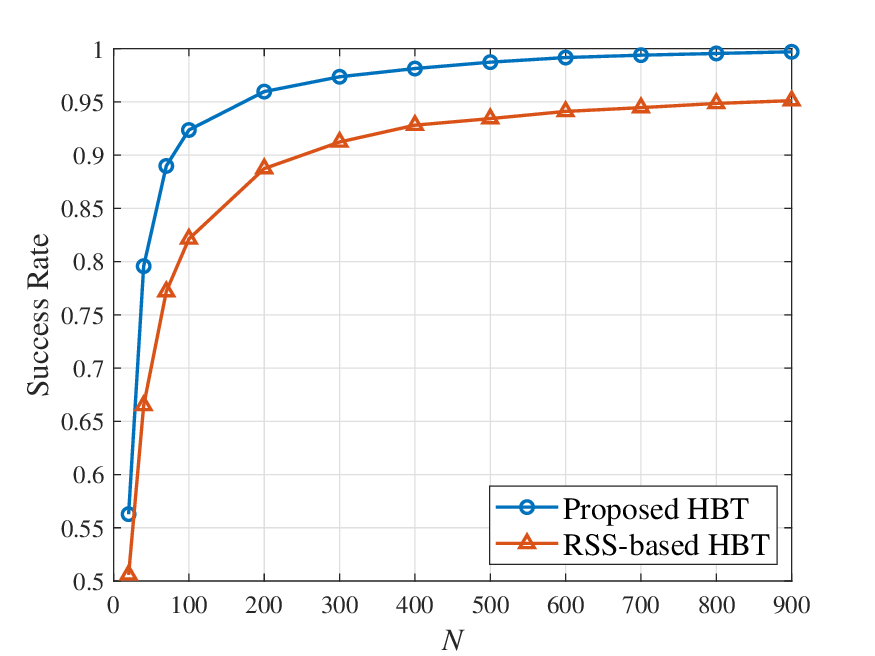}
  \caption{Target detection success rate versus $N$.}
  \label{CS_N}
\end{figure}
\color{black}
\par Then, we compare the target direction estimation error of the proposed HBT scheme with four benchmark schemes: 1) HBT without the BR process; 2) HBT without the BC process; 3) EBS scheme \cite{9547829}; 4) MLE-based random beam training (MLE-RBT) scheme \cite{WangJoint2}, with its beam training overhead set to be the same as the proposed HBT scheme (i.e., $( 2+4\log _2M ) $ OFDM symbols). The target direction estimation error is described by
\begin{align}
\varepsilon _{\mathrm{DR}} =\mathbb{E} \left\{ \pi \sqrt{\left( u_{\mathrm{R}2\mathrm{G}}^{\mathrm{D}}-\hat{u}_{\mathrm{R}2\mathrm{G}}^{\mathrm{D}} \right) ^2+\left( v_{\mathrm{R}2\mathrm{G}}^{\mathrm{D}}-\hat{v}_{\mathrm{R}2\mathrm{G}}^{\mathrm{D}} \right) ^2} \right\} .
\end{align}
As shown in Fig.~\ref{CS_M}, with fewer IRS elements, the proposed HBT scheme underperforms the EBS scheme and the MLE-RBT scheme, due to the inadequate beamforming gain at the beginning stages of the HBT. However, by adopting more IRS elements, the direction estimation accuracy of the proposed HBT scheme gradually surpasses that of the MLE-RBT scheme, because the random beam training can not fully exploit the IRS beamforming gain and the Doppler shift effect limits the performance of MLE. Subsequently, the proposed HBT outperforms the EBS scheme, attributing to the use of the BR process that further improves the direction estimation accuracy. Moreover, the target direction estimation accuracy improves across all training schemes as the quantity of IRS elements increases, This enhancement is because, adding IRS elements not only improves the codeword resolution (i.e., $2/M$), but also obtains more substantial IRS beamforming gain to benefit the target detection process. In addition, the exclusion of either the BR or BC process results in a deterioration in HBT performance, highlighting the indispensable role of the BR/BC process.
\begin{figure}[htbp]
  \centering
   \includegraphics[width=2.8in]{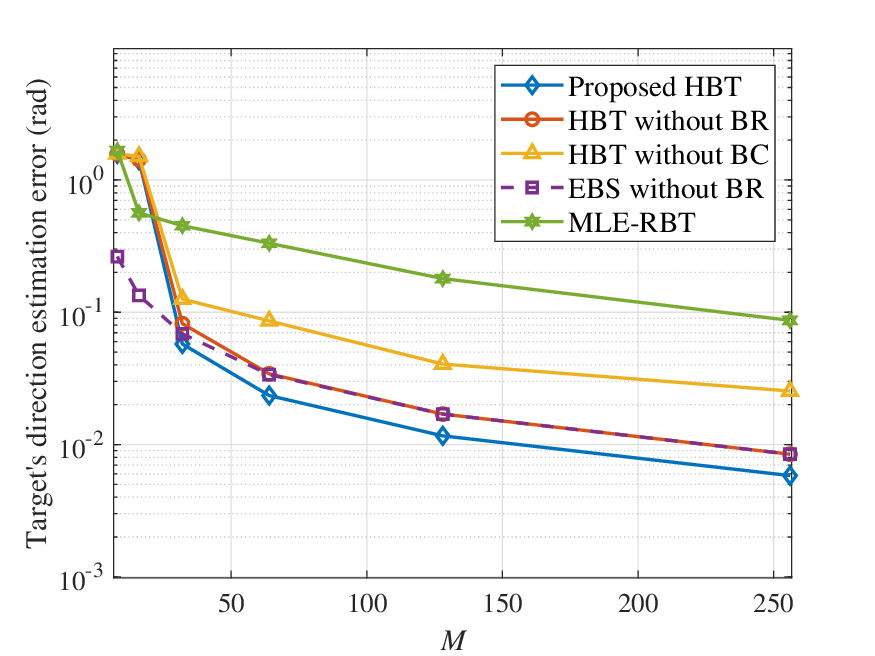}
  \caption{Target direction estimation error versus $M$.}
  \label{CS_M}
\end{figure}
\color{black}

\subsection{Performance of target R\&V Estimation in the FGS Period}
\color{black}
\begin{figure}[htbp]
  \centering
  \subfigure[Range Estimation.]  
  {
  \includegraphics[width=2.8in]{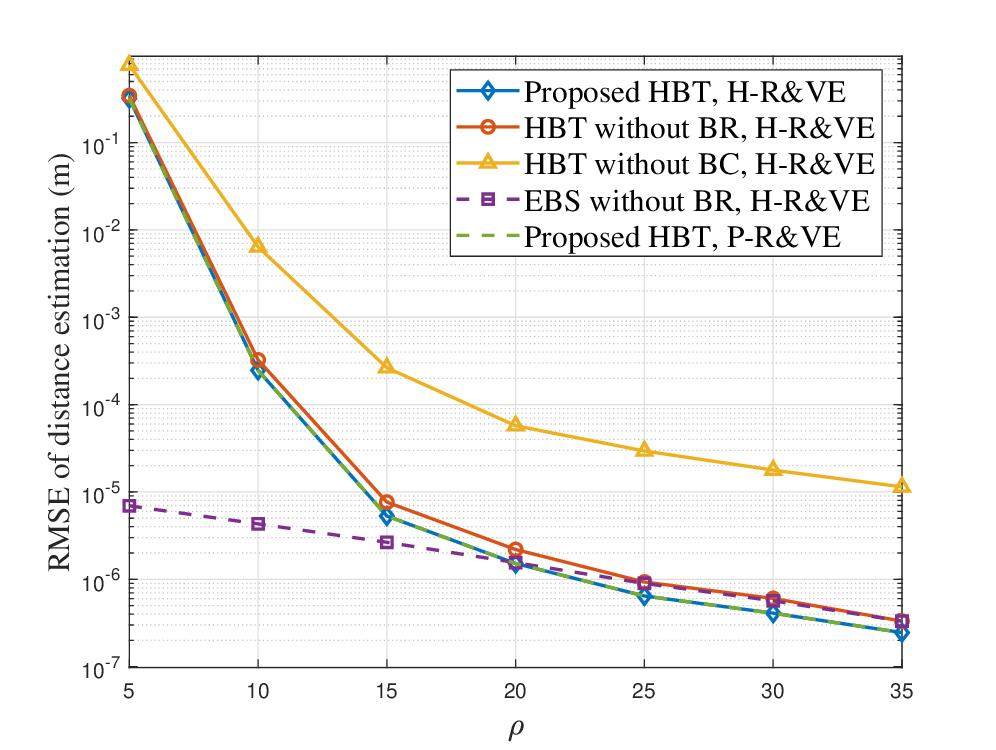}
  }
  \subfigure[Velocity Estimation.] 
  {
  \includegraphics[width=2.8in]{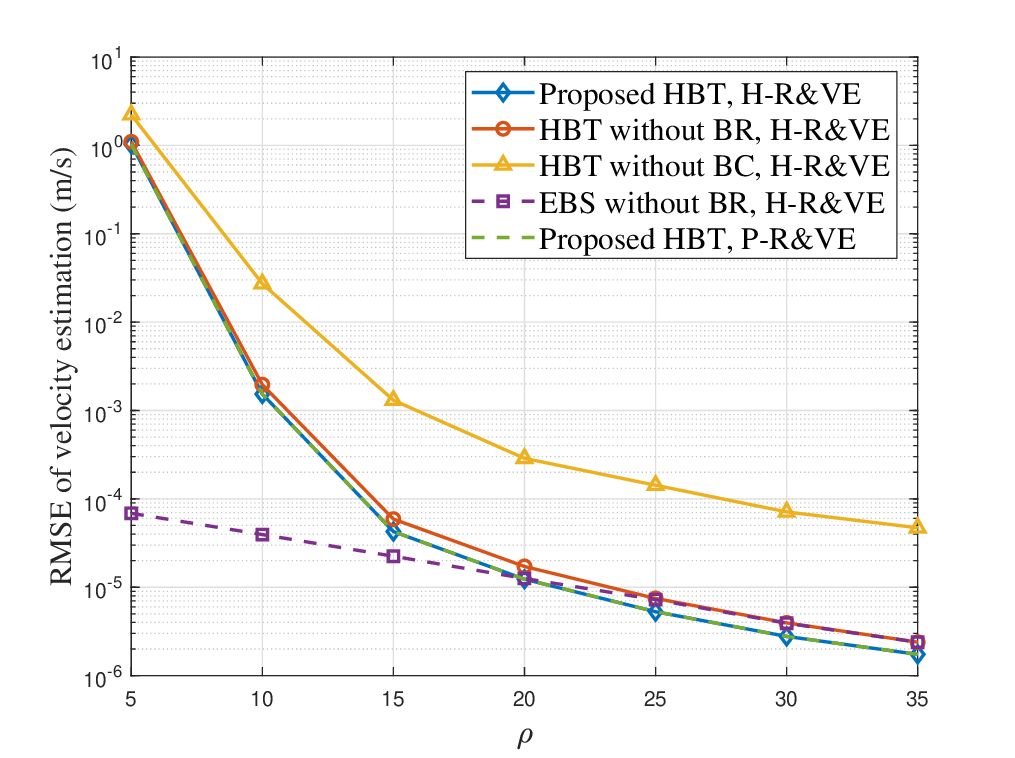}
  }
  \caption{R\&V estimation accuracy with different beam training schemes.}
  \label{FS_rho}
\end{figure}

Fig.~\ref{FS_rho} provides comparisons between the FGS performance across different beam training schemes and different R\&V estimation schemes, where “P-R\&VE” refers to the periodogram-based R\&V estimation algorithm \cite{PucciSystem2022}. The accuracy of R\&V estimation is evaluated by the root mean square error (RMSE) defined as
\begin{align}
\varepsilon _{\mathrm{D}}\!=\!\sqrt{\mathbb{E} \left\{ \!\left| \hat{d}_{\mathrm{R}2\mathrm{G}}-d_{\mathrm{R}2\mathrm{G}} \right|^2\! \right\}},\varepsilon _{\mathrm{V}}\!=\!\sqrt{\mathbb{E} \left\{ \!\left| \hat{\mathrm{v}}_{\mathrm{G}}-\mathrm{v}_{\mathrm{G}} \right|^2\! \right\}}.
\end{align}
The proposed H-R\&VE algorithm achieves the same estimation accuracy as the P-R\&VE algorithm while reducing the required calculations of DD spectrum elements from $\left( N_{\mathrm{R}}T_{\mathrm{V}} \right) ^I$ to $\left( 4I+2 \right) \sqrt[I]{N_{\mathrm{R}}T_{\mathrm{V}}}$.
Since the configuration of IRS beamforming is directly influenced by the outcome of the beam training process, the performance of both R\&V estimations varies in accordance with the adopted beam training schemes. At the lower SNR level, the EBS excels with the highest accuracy of R\&V estimation, benefiting from its stable direction estimation ability. Nevertheless, at the medium/high SNR level, the proposed HBT scheme performs the best in terms of both R\&V estimation, attributing to the BR process adoption. Moreover, excluding the BC process results in a profound degradation of the accuracy in both R\&V estimation, which verifies the pivotal role of the BC process in FGS accuracy.
\color{black}

\begin{figure}[htbp]
  \centering
  \subfigure[Range estimation accuracy versus $N$.]  
  {
  \label{FS_N}
  \includegraphics[width=2.8in]{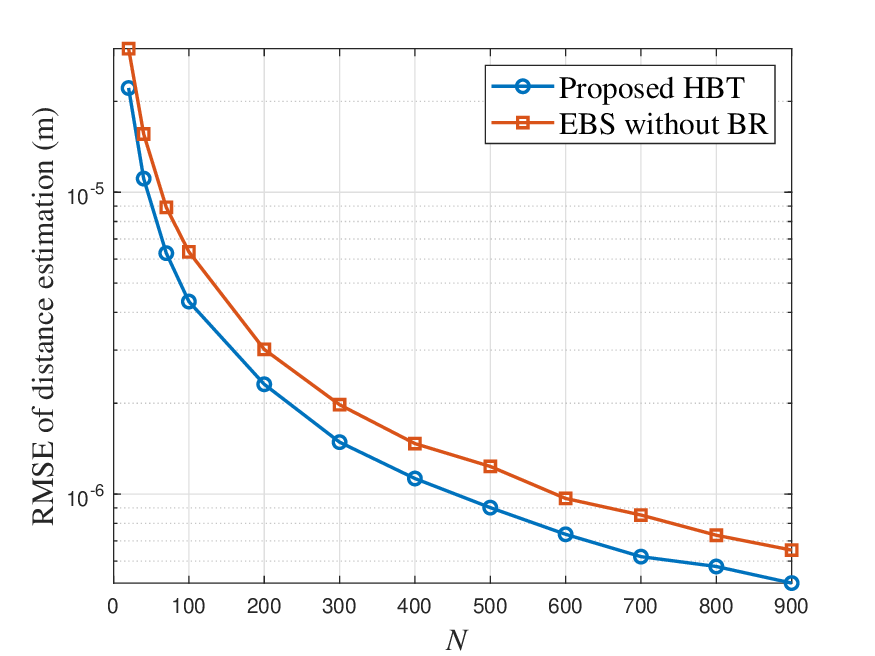}
  }
  \subfigure[Velocity estimation accuracy versus $T_{\mathrm{FG}}$.] 
  {
  \label{FS_T}
  \includegraphics[width=2.8in]{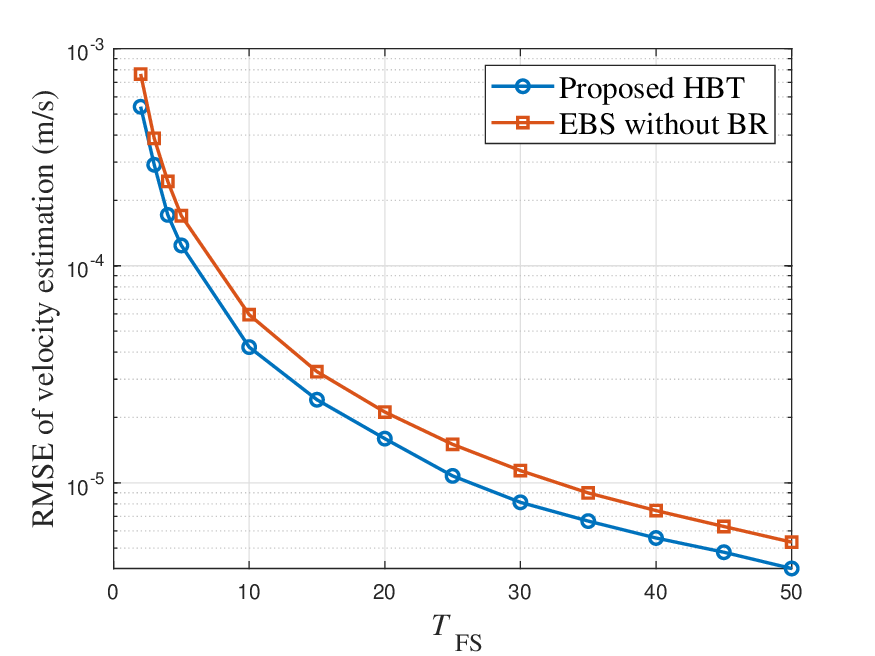}
  }
  \caption{R\&V estimation accuracy under different time-frequency resource allocation conditions.}
  \label{FS_NT}
\end{figure}
Fig.~\ref{FS_NT} presents the influence of time-frequency resource allocation on the performance of FGS. It can be seen that adopting more SCs and OFDM symbols improves range estimation accuracy and velocity estimation accuracy, respectively. As expected in (\ref{DDSpectrum}), collecting observation data from more SCs and OFDM symbols helps decrease the noise power in the DD domain, thereby increasing the reliability of R\&V measurement.

\begin{figure}[htbp]
  \centering
  \subfigure[Range Estimation.]  
  {
  \includegraphics[width=2.8in]{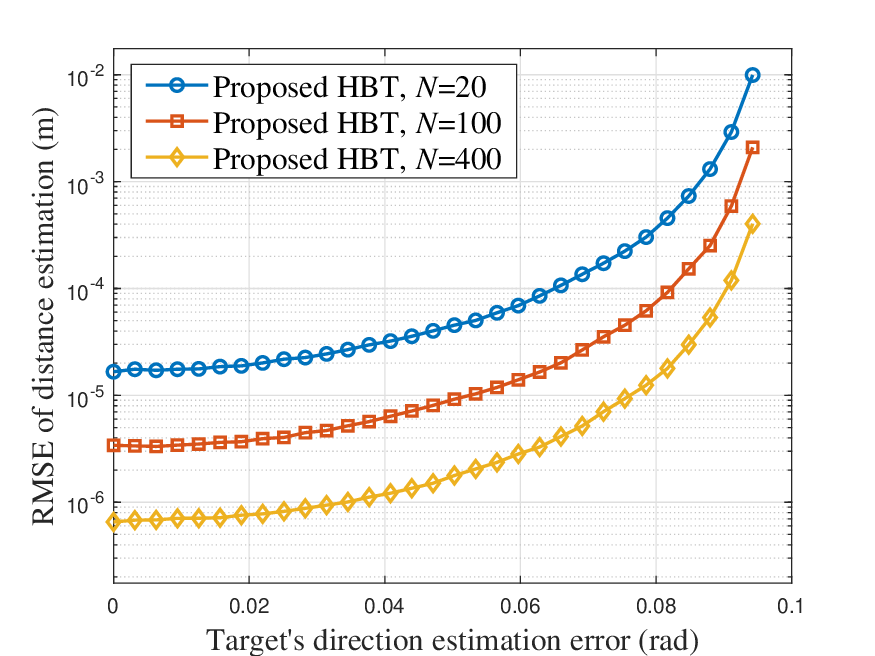}
  }
  \subfigure[Velocity Estimation.] 
  {
  \includegraphics[width=2.8in]{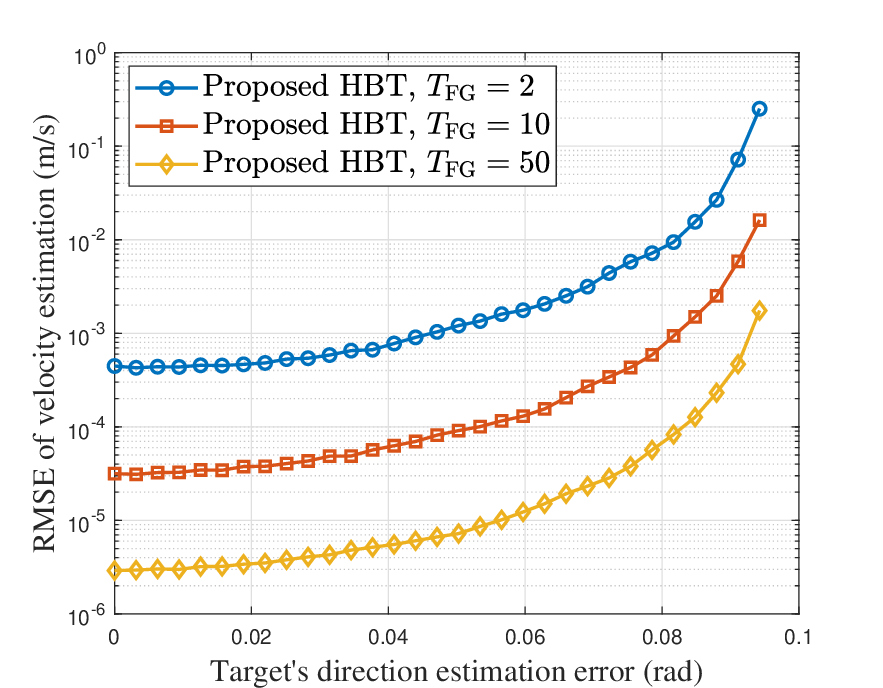}
  }
  \caption{R\&V estimation accuracy versus target's direction estimation error.}
 \label{CtoF}
\end{figure}
Fig.~\ref{CtoF} illustrates the impact of CGS performance on the FGS performance. Since the IRS beam points at the target's direction estimated in the CGS period, the accuracy of the target R\&V estimation in the FGS period increases as the direction estimation accuracy becomes higher. Specifically, with the direction estimation accuracy increasing, the FGS performance improves slightly in the high-accuracy region of direction estimation, while improving significantly in the low-accuracy region. In addition, when small time-frequency resources are allocated for FGS, the resultant decrease in its performance can be counterbalanced by enhancing the CGS performance. For instance, when the number of SCs decreases from $400$ to $100$, the range estimation accuracy remains unchanged at $10^{-5}$ m by decreasing the target's direction estimation error from about $0.08$ rad to $0.05$ rad.

\begin{figure}[htbp]
  \centering
   \includegraphics[width=2.8in]{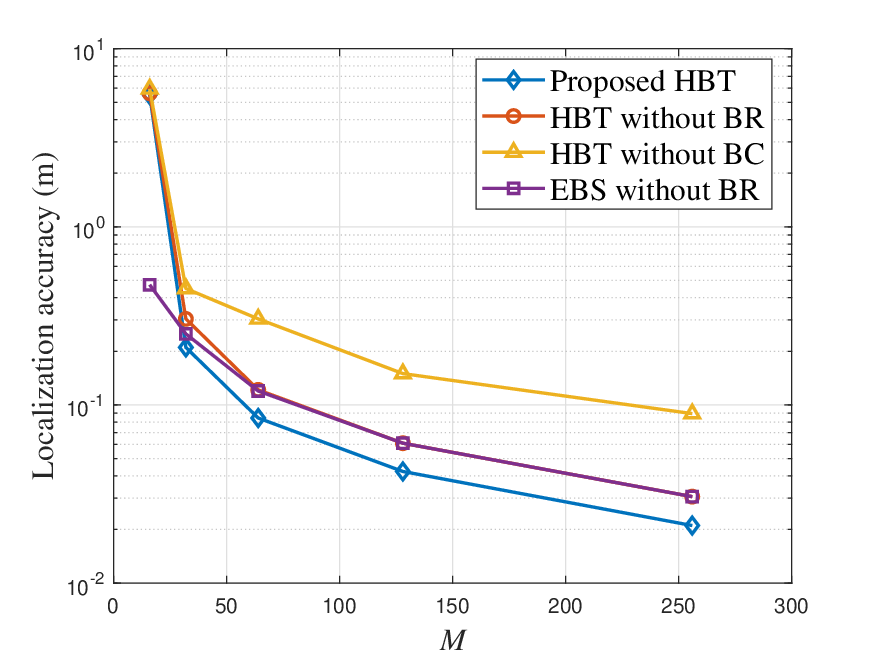}
  \caption{Localization accuracy versus $M$.}
  \label{localization}
\end{figure}
Finally, Fig.~\ref{localization} studies the localization accuracy of the proposed target sensing system versus the numbers of IRS elements across various beam training schemes, where the localization accuracy is evaluated by the RMSE defined as
\begin{align}
\varepsilon _{\mathrm{L}}=\sqrt{\mathbb{E} \left\{ \left| \hat{\mathbf{q}}_{\mathrm{G}}-\mathbf{q}_{\mathrm{G}} \right|^2 \right\}}.
\end{align}
The localization accuracy with the proposed HBT strategy is lower than with the EBS strategy at first, and gradually becomes better
with the increase of the number of IRS elements. Without either the BR or BC process, the localization accuracy gets worse. This is because, more accurate direction estimation in the CGS period can provide higher beamforming gain in the FGS period to achieve better range estimation performance, while the localization performance is exactly determined jointly by the direction estimation accuracy and range estimation accuracy.

\section{Conclusion}\label{section7}
In this article, we considered the NLOS target sensing problem in an IRS-assisted OFDM DFRC system, and designed a novel target sensing framework, including the sensing protocol, detection and direction estimation scheme, as well as R\&V estimation scheme. Specifically, we first designed a target sensing protocol, where the whole CPI is divided into the CGS period for target detection and direction estimation and the FGS period for target R\&V estimation. \black{For the CGS period, we proposed a low-overhead hierarchical codebook for IRS 3D beamforming, developed a DSP-based HBT scheme for target detection and direction estimation, and designed a BR scheme for further enhancing the direction estimation accuracy.} For the FGS period, we estimated the target's R\&V by extracting the DD information from the IRS relected echo signals. Numerical results showed that, the proposed DSP detector can reach a remarkable detection success rate of $99.7\%$, outperforming the traditional power detector by $10\%$. \black{Simulations also indicated that the proposed IRS 3D beam training scheme suits the scenario with a substantial quantity of IRS elements, where the proposed scheme can precisely estimate the target direction with $10^{-3}$-rad level accuracy, even surpassing the EBS scheme, which is attributed to the usage of the BR process and the DSP detector.} In addition, by applying the beam training result to assist the FGS period, a $10^{-6}$-m level range estimation accuracy and a $10^{-5}$-m/s level velocity estimation accuracy can be achieved.

\bibliographystyle{IEEEtran}
\bibliography{references}{}

\end{document}